\documentclass[12pt]{article}
\pdfoutput=1 
\usepackage{makeidx}
\makeindex

\title{A local-realistic model for quantum theory}
\author{ Paul Raymond-Robichaud \\
{\normalsize ISI Foundation, Torino, Italy}\\ 
{\normalsize paul.r.robichaud@gmail.com}
}

\usepackage[breaklinks]{hyperref}
  \usepackage{verbatim}
 \usepackage{lipsum}

\usepackage{amsthm}
\usepackage{times}
\usepackage{amsmath,amssymb}
\usepackage[english]{babel}

\usepackage{tikz}

\usepackage{url}
\usetikzlibrary{decorations.pathreplacing}

\input{Qcircuit}

\theoremstyle{definition}
\newtheorem{definition}{Definition}[section]

\newtheorem{theorem}{Theorem}[section]

\newtheorem{corollary}{Corollary}[section]

\newtheorem{axiom}{Axiom}[section]

\newcommand{\iffdef}{\stackrel{\mathrm{def}}{\iff}}

\newcommand{\isdef}{\stackrel{\mathrm{def}}{=}}

\newlength{\spacingA}
\newlength{\spacingB}
\newlength{\spacingC}
\begin{document}

 \maketitle

\begin{abstract} We provide a rigorous definition of local realism.
 We show that the universal wave function cannot be a complete description of a local reality. Finally, we construct a local-realistic model for quantum theory.
\end{abstract}

\section*{Acknowledgements}

I am   beyond grateful to Gilles Brassard  for countless stimulating discussions, many years of supervision, a thorough review of a draft of this article, along with numerous helpful suggestions.  I am also thankful for his unwavering encouragement and support.  Gilles   has been the main voice to spread far and wide the ideas presented here.  The numerous conferences that he has given on the topic were  more eloquent than anything I  could have done. This work would not have been possible without him.

I am grateful to Chiara Marletto,  David Deutsch   and Charles Alexandre Bédard for countless stimulating discussions and for providing numerous suggestions for improvements to a draft of this article.

I am grateful  to Lev Vaidman for  many   helpful suggestions to improve the present article.

I am grateful  to  Stefan Wolf, Tony Short, Renato Renner, Mario Rasetti, Sandu Popescu,  Marcin Paw{\l}owski,  Patrick Hayden,  St\'ephane Durand,  Jeff Bub, Michel Boyer and last but not least Charles Bennett for many stimulating discussions.

This work has been supported in part by NSERC, the Fonds de recherche du Qu\'ebec -- Nature et technologies and from Intesa Sanpaolo Innovation Center. The funders had no role in study design, data collection, and analysis, decision to publish, or preparation of the manuscript.

\section{Introduction}
Quantum theory has long been held to be incompatible with  local realism. 
   Local realism can be described informally, albeit incompletely, by the following principles:

\paragraph{Principles of Local Realism} 
\begin{enumerate}
\item There is a real world.
\item It can be decomposed into various parts, called \emph{systems}.
\item A system may be decomposed into subsystems.
\item  Every system is a subsystem of the \emph{global system} consisting of the entire world.
\item At any given time, each system is in some state.
\item The state of a system  determines, and is determined by, the state of its subsystems. 
\item What is observable in a system is determined by the state of the system.
\item The state of the world evolves according to some  law.
\item  The evolution of the state of a system can only be influenced by the state
  of systems in its local neighbourhood.
\end{enumerate}

The notion of local neighbourhood depends on the underlying physics. In the case of relativity theory, this would mean that no systems  can influence each others if they are space-like separated.  This implies that no action on a system can have an influence on  another system at a speed faster than light. However, these principles are not restricted  to  theories in which general relativity holds.

Two reasons are usually given as to why quantum theory cannot be compatible with local realism.

The first  is that quantum theory is incompatible with local hidden variable theories \cite{Bell64}. But those are merely a specific type of local-realistic theory. And while all local hidden variable theories are local-realistic, the converse is false.  For example, the ``non-local box'' introduced by Sandu Popescu and Daniel Rohrlich cannot be described by local hidden variables and yet can easily be given a local-realistic interpretation \cite{PR,poster,ParLives}. Note that the principles of local realism that we have enumerated make no reference to local hidden variables.

The second  is that an entangled state of a composite system cannot be described in term of the states of its subsystems.  This, if true,  would indeed  imply that quantum theory is not  local-realistic. However, we shall show precisely that in quantum theory, all states of all composite systems, including entangled states, can be fully described by the  states of their subsystems.  

Despite these objections that are still made today, the consistency of quantum theory with local realism has already been proven by David Deutsch  and Patrick Hayden \cite{DH}.  They appealed  to   an intuitive conception of local realism that is in full agreement with the formal framework provided in this article. They also provided a local-realistic model for quantum theory.
The model  presented here can easily be shown to be isomorphic with theirs.  Charles Alexandre Bédard provided a comparison of both models and showed their equivalence in his thesis \cite{CAB, CAB2}. 

But our aim here is more general.  We shall present a precise formulation of the principles of local realism, we can have precise theorems of an extremely general nature, which deals not just 
with quantum theory but will \emph{all} possible local-realistic theories.
  One example  is the no-signalling principle, which is valid in all local-realistic theories, and generalizes a well-known result of quantum theory. The no-signalling principle is that no operation on a system can have an instantaneous \emph{observable} effect on  remote systems.  That the no-signalling principle holds in all local-realistic theories is proven in a companion article, where we also show that every  no-signalling theory with reversible dynamics has a local-realistic model \cite{STRUCTURE}.  Given that all local-realistic theories are  no-signalling, this implies the following theorem: given  reversible dynamics, a theory is no-signalling if and only if it has a local-realistic model. 

Developing those ideas required a  formalization of the concept of local realism  \emph{independent} of quantum theory. Developing a local-realistic model for an  arbitrary no-signalling  theory was obtained by abstracting away all that belongs specifically to quantum theory from the model in the current article and retaining only the most general mathematical ideas.

While the model provided in the companion article is universal, the price of universality is at  a higher degree of abstraction than in the present particle, which is built around the familiar density operators, vector spaces and Hilbert spaces of quantum theory.   For a comparison of both models, see  \cite{CAB, CAB2}.  

The local-realistic model presented in this article is valid for finite-dimensional unitary quantum theory.  
For the sake of brevity, in the remainder of this article,  we shall often refer to finite-dimensional unitary quantum theory as  quantum theory.
The model of the current article can be extended to continuous-dimensional unitary quantum theory as shown in \cite{CAB, CAB2}.

This paper is an extension of the doctoral work of the author, performed under the supervision of Gilles Brassard. It is intended to supersede previous work by the author on the topic \cite{PRRthesis,pres,NDFP15}. In this paper, as opposed to previous work, local realism is described through explicit axioms, and greater emphasis is placed on explaining how quantum theory should satisfy these axioms. Furthermore, the local-realistic model for quantum theory and the proof that the universal wave function cannot be a complete description of local reality are presented here with full rigour and go in greater detail than the previous presentations. Nevertheless, all the key ideas were already presented by the author in a conference he gave in 2014~\cite{pres}, including the main axioms of local realism (called ``desiderata'' at the time) and the model based on \emph{evolution matices} (Definition~\ref{evoma}). 
Note that the slides show the author's supervisor as first author, yet all the ideas as well as the presentation itself originate with the author. A preliminary statement and proof that the universal wave function is but a shadow of the real world (under the metaphysical assumption of locality) was first presented by the author's supervisor at the 2015 edition of the annual \emph{New Directions in the Foundations of Physics} conference organized by Jeffrey Bub on slides 38 to 40 of ref.~\cite{NDFP15} with the author's permission. However, this proof was entirely the discovery of the author even though he enjoyed listening to his supervisor serving as his voice.

\section{The density operator of a system does not determine the state of the system} \label{qtnotlr}
All observable properties of quantum theory can be described by the density operator formalism.   
A system's density operator determines to everything that is locally observable on it.  
However, the density operator of system $A$ and system $B$, do not determine the density operator of composite system $AB$.
Yet, in a local-realistic theory, the state of composite system \emph{is}   determined by the state of its subsystems. 

We could be tempted to  jump to the conclusion that  quantum theory is not local-realistic.  However this would only follow   under the supposition that the state of a system is determined by its density operator.   

This argument, rather than proving that quantum theory is not local-realistic gives us a  hint for how  to find a local-realistic model for quantum theory.  Namely, by contrapositive we find that  any local-realistic model of quantum theory must distinguish between the state of a system and  what is locally observable in it --  which according to the principles of local-realism,  is determined by the  state of the system.

\section{Local-realistic theories} \label{sc:locreltheo}
Hence, our next task is to give a mathematical definition  of a local-realistic theory.     A local-realistic theory is defined as a mathematical theory that satisfies \emph{all} the axioms we shall present.  

Mathematics alone is insufficient and powerless to apprehend reality.  To fully understand a physical  concept such as local realism, there is by necessity an intuitive level, where the key consideration is the relation between mathematics and reality.  This level by its nature cannot be discussed with the same degree of rigor.
We have discussed part of the intuitive meaning of what is a local-realistic theory in the introduction of this paper.   A deeper exposition of non-mathematical aspects of local realism and  the relevant mathematics  is given in \cite{STRUCTURE}.

\subsection{Systems}
We want to define the mathematical properties of systems, where systems describe meaningful parts of the world.  These systems can be combined in various way to give rise to different systems.  For example, if $A$ is a system, there will be a system $\overline{A}$, the complement of system $A$, consisting of the rest of the world.  As an other example, if $A$ and $B$ are systems, there will be a system $A \sqcup B$, the union of system $A$ and $B$ consisting of the parts of the world consisting of the part of the world belonging to  either of the two systems $A$ and $B$.

Mathematically,  systems will be represented as elements of a lattice of systems, which we define now.

\begin{definition}[Lattice of systems]
A \emph{lattice of system} is a  6-tuple $\left( \mathcal{S}, \sqcup, \sqcap, \overline{\,\cdot\,}, S, 0 \right)$, where
 $\mathcal{S}$ is a set of elements called \emph{systems}, symbols  $\sqcup$, $\sqcap$ denote binary operations on the set of systems $\mathcal{S}$, symbol  $\overline{ \, \cdot \, }$ denotes a unary operation on the set of systems $\mathcal{S}$, symbol   $0$ and $S$ are distinguished elements of the set of systems $\mathcal{S}$ respectively called the \emph{empty system} and the \emph{global system}.

If $A$ and $B$ are systems,  we shall say that: 
\begin{enumerate}
\item $A \sqcup B $ is the \emph{union} of $A$ and $B$, 
\item $A \sqcap B $ is the \emph{intersection} of $A$ and $B$,
\item $ \overline{A}$ is the \emph{complement} of $A$.
\end{enumerate}
The following properties are verified for all systems $A$, $B$ and $C$:
\begin{enumerate} 
\item Associativity:
\begin{align*}A \sqcup ( B \sqcup C ) &= ( A \sqcup B ) \sqcup C \, , \\
A \sqcap ( B \sqcap C )  &= ( A \sqcap B ) \sqcap C  \, .
\end{align*}
\item Commutativity:
\begin{align*}
 A \sqcup B &= B \sqcap A \, , \\
 A \sqcap B  & = B \sqcap A \, . \\
\end{align*}
\item Absorption:
\begin{align*}
 A \sqcap ( B \sqcup A ) &= A \, , \\
 A \sqcup ( B \sqcap A ) &= A \, .
\end{align*}
\item Distributivity:
\begin{align*}
A \sqcap ( B \sqcup C ) &= ( A \sqcap B ) \sqcup ( A \sqcap C) \, ,  \\
A \sqcup ( B \sqcap C ) &= ( A \sqcup B ) \sqcap ( A \sqcup C ) \, .
\end{align*}
\item Complementation:
\begin{align*}
 A \sqcup \overline{A} &= S \, , \\
 A \sqcap \overline{A} &= 0  \, .
\end{align*}
\end{enumerate}
\end{definition}
 Formally,  a lattice of system is simply  a  boolean lattice.  However, we believe that an evocative terminology helps to build proper intuition, thus the properties of boolean lattices  are better described using the terminology of physical systems, rather than the general but flavourless language of boolean lattices.

Conceptually the global system could represent as much as the entire universe, or a much smaller system.  Though our goal is ultimately to be able to represent as much as possible, we might want to describe in a local-realistic way smaller systems like a quantum computer or several qubits.

Given the definition of a  lattice of systems, we can state the first axiom.

\begin{axiom}[Systems]
Associated to a local-realistic theory is a  lattice of systems  $\left( \mathcal{S}, \sqcup, \sqcap, \overline{\,\cdot\,}, S, 0 \right)$.
\end{axiom}

We now introduce some terminology on systems.

There is  a natural partial order ``$\sqsubseteq$'' among systems  defined  by:
\[ B \sqsubseteq A \iffdef  A \sqcap B = B \, . \]

\begin{definition}[Subsystem]  A system $B$ is a \emph{subsystem} of system $A$ if $ B \sqsubseteq A$.  
\end{definition}

\begin{definition}[Proper subsystem] A system $B$ is a \emph{proper subsystem} of a system  $A$ if  $B$ is a subsystem of $A$ and $B$ is distinct from the empty system $0$  and system $A$.
\end{definition}

\begin{definition}[Atomic system]  A   system  $A$  is \emph{atomic} if it has no proper subsystem and it is distinct from the empty system  $0$.
\end{definition}

\begin{definition}[Disjoint systems] Two systems $A$ and $B$  are \emph{disjoint} if  $A \sqcap B = 0 $.
\end{definition}

\begin{definition}[Non-trivial system] A  system $A$ is non-trivial, if it is neither the empty system $0$, nor the global system $S$.
\end{definition}

\begin{definition}[Composite system]
Whenever systems $A$ and $B$ are disjoint, we say that the system $A \sqcup B$ is a \emph{composite system}, composed of system $A$ and system  $B$.  For convenience, we denote it by $AB$, rather than $A\sqcup B$. 
\end{definition}

Since $\sqcup$ is commutative, we have 
\mbox{$AB = A \sqcup B = B \sqcup A = BA$}.

Since $\sqcup$ is also associative, we have
\mbox{$A \left( BC \right) = \left( A B \right) C$}
for any three mutually disjoint systems $A$, $B$ and $C$.
Thus, we shall  write $ABC$ to denote the composite system consisting of mutually disjoint systems $A$, $B$ and $C$.

\begin{definition}[Finite lattice of systems] A lattice of systems $\left( \mathcal{S}, \sqcup, \sqcap, \overline{\,\cdot\,}, S, 0 \right)$ is finite if its set of systems $\mathcal{S}$ is finite.
\end{definition}

\begin{definition}[Finite union of systems]
If  $\mathcal{A}$ is a finite set of systems,
we define recursively  $ \bigsqcup\limits_{A \in \mathcal{A} } A  $,  the \emph{union} of $ \mathcal{A}$,  in the following way:
\[\bigsqcup\limits_{A \in \mathcal{A} } A  \isdef \begin{cases} \, 0 \text{, the empty system, } & \text{ if } \mathcal{A}  \text{ is the empty set, } \\
                   ( \bigsqcup\limits_{ A \in  \mathcal{A} - \{ B \}  }   \! \! \! \! \! \!  A ) \sqcup B  & \text{ if }  B \text{ is any element  of } \mathcal{A} \, .
\end{cases}\] 
\end{definition}

It can be verified that $\bigsqcup\limits_{A \in \mathcal{A} } A$  is a system.
It can further be verified that for all  finite sets of systems $\mathcal{A}$ and $\mathcal{B}$:
\[ \bigsqcup\limits_{A \in \mathcal{A} \cup \mathcal{B} } A = ( \bigsqcup\limits_{A \in \mathcal{A} } A) \; \sqcup \; ( \bigsqcup\limits_{A \in \mathcal{B} } A) \]
We could have defined the finite intersection in the same way, but shall not need to do so for the purpose of this article.

The following theorem is a well-known theorem of finite boolean lattices, expressed in the terminology of systems. 
\begin{theorem} Let $ ( \mathcal{S} , \sqcup , \sqcap ,\overline{\cdot} ,  S, 0 )$ be a finite lattice of systems,   let $B$ be a system, there exist a unique  set of atomic systems $\mathcal{A}$ such that $B = \bigsqcup\limits_{A \in \mathcal{A} } A \, . $
\end{theorem}

We shall refer to $ \mathcal{A} $ as the \emph{atomic decomposition} of   $B$.

We shall soon discuss properties of systems in quantum theory, and for that, we will need to  work with the tensor product of Hilbert spaces. To fix some notation, we first  define  what is meant  by a tensor product.  

\begin{definition}[Tensor Product]  Let $\mathbb{F} $ be a field, let $V_{1} $, $V_{2}$, $V_{3}$ be vector spaces over the field $\mathbb{F} $ and  let $ \otimes \colon V_{1} \times V_{2} \to V_{3}$ be a bilinear operation.  We say that $( V_{3} , \otimes ) $ is a \emph{tensor product} of $V_{1}$ and $V_{2} $ if for all vector spaces $W$ over the field $\mathbb{F}$, and all bilinear functions $f \colon V_{1} \times V_{2} \to W $, there exists a unique linear application $g \colon V_{3} \to W$ such that for all vectors $v_{1}$ of  $V_{1}$ and all vector $v_{2} $ of  $V_{2}$, the following holds:
\[  f( v_{1}, v_{2} ) = g ( v_{1} \otimes v_{2} ) \, . \]
\end{definition}
If $ ( V_{3} , \otimes ) $ is a tensor product of $V_{1}$ and $V_{2}$, we write $ V_{3} = V_{1} \otimes V_{2} $ and we also refer to operation  $ \otimes$  as a \emph{tensor product}.

With slight modifications, this is the category-theoretic definition of a tensor product.

\paragraph{Systems in finite dimensional unitary quantum theory}
When the local-realistic theory is finite unitary quantum theory, its associated lattice of systems is finite, and  associated to each system $A$ is a finite dimensional complex Hilbert space denoted $ \mathcal{H}^{A}$.  Furthermore, if system $A$ is distinct from the empty system, its associated Hilbert space $\mathcal{H}^{A}$ is of dimension larger than one.  This implies that given a non-empty system $A$, the Hilbert space $\mathcal{H}^{A}$ associated to system $A$ always  contains at least two orthonormal vectors. Thus a non-empty system $A$ contains at least a qubit of information.

Also, associated with every pair of disjoint system $A$ and $B$ is a tensor product $\otimes_{( A, B)}$ such that :
\[ \mathcal{H}^{AB} = \mathcal{H}^{A} \otimes_{(A, B)} \mathcal{H}^{B} \, .
 \] 
Moreover, the family of  tensor products  satisfies the two following properties for any arbitrary mutually disjoint system $A$, $B$ and $C$, and any arbitrary vectors $ | u \rangle^{A} $, $| v \rangle^{B} $, $ | w \rangle^{C}$ belonging respectively to the Hilbert spaces of systems $A$, $B$, and $C$:
\[   | u \rangle^{A} \otimes_{(A, B)} | v \rangle^{B} = | v \rangle^{B} \otimes_{( B,A)} | u \rangle^{A} \, , \]
and
\[  ( | u \rangle^{A} \otimes_{( A, B) } | v \rangle^{B} ) \otimes_{( AB,C)} | w \rangle^{C} =  | u \rangle^{A}  \otimes_{( A, BC)} ( | v \rangle^{B}  \otimes_{( B, C)} | w \rangle^{C} )  \, . \]
These properties are consistent with the fact that since  system $AB$ is equal to system $BA$, their underlying Hilbert spaces are identical, and also that since  system $(AB)C$ is identical with system $A( BC)$ their underlying Hilbert spaces are also  identical.
When there is no ambiguity, we shall drop the subscripts and write $\otimes$  to denote all tensor products. For example, the two lines above shall be written as:
\[   | u \rangle^{A} \otimes | v \rangle^{B} = | v \rangle^{B} \otimes | u \rangle^{A} \, , \]
and
\[  ( | u \rangle^{A} \otimes | v \rangle^{B} ) \otimes | w \rangle^{C} =  | u \rangle^{A}  \otimes ( | v \rangle^{B}  \otimes | w \rangle^{C}  ) \, . \]
Since the tensor product behaves associatively, as an abuse of notation we shall omit the unnecessary parentheses and write $  | u \rangle^{A} \otimes | v \rangle^{B}  \otimes | w \rangle^{C} $ to denote the tensor product of $ | u \rangle^{A}$, $ | v \rangle^{B} $ and $ | w \rangle^{C}$.

\subsection{States}
We now turn our attention to  the states of systems.
As we discussed previously in section  \ref{qtnotlr} we need to distinguish between what can be  locally observable in a system and its complete description.    This leads us to the distinction between two kinds of states for a system.

\begin{description}
\item The \textbf{noumenal state} of a system is its complete  \emph{local} description.  
\item The \textbf{phenomenal state} of a system is a complete description of what is \emph{locally} observable in that system.  
\end{description}

The ``noumenal/phenomenal'' terminology is inspired by Kant's;  however our usage  should neither be taken as endorsement of Kantian metaphysics nor as claim  that our usage conforms to his  \cite{Kant,Kant-Brit}.

This leads us to corresponding  axioms:

\begin{axiom}[Noumenal state space]
Associated to a system $A$ is  \emph{noumenal state space}, $\textsf{Noumenal-Space}^{A}$,  which is a set of \emph{noumenal states}.
\end{axiom}

 Particular noumenal states of $A$ will be  denoted $N^{A}$,  $ N^{A}_i$, $ N^{A}_1$, etc. 

Intuitively, the noumenal state space of a system consists of all possible noumenal states that it could theoretically be in.

\begin{axiom}[Phenomenal state space]\label{pheno}
Associated to a system $A$  is a \emph{phenomenal state space}, $\textsf{Phenomenal-Space}^{A}$, which is a set of \emph{phenomenal states}.
\end{axiom} 

Particular phenomenal states of $A$ will be denoted $\rho^{A}$, $\rho^{A}_{i}$, $\rho^{A}_{1}$, etc.

Intuitively, the phenomenal state space of a systems consist of all possible phenomenal state that it could theoretically be in.

\paragraph{States in  quantum theory. } 
In quantum theory, given a Hilbert space $\mathcal{H}$, we will denote the set of density operators on the Hilbert space $\mathcal{H} $ as 
\[ \textsf{Density-Operator}^{\mathcal{H}} = \{ \rho \, \colon \, \rho \text{ is a density operator on } \mathcal{H} \} \, . \] 
In quantum theory the phenomenal state space of a system $A$, with associated Hilbert space $\mathcal{H}^{A}$  is
\[ \textsf{Phenomenal-Space}^{A} = \textsf{Density-Operator}^{\mathcal{H}^{A}} \, . \]
However, the density operator formalism of quantum theory does not supply noumenal states. Therefore noumenal states will need to be given later in order to provide quantum theory with a local-realistic model.

\paragraph{Pure state quantum theory}
A variant of quantum theory, which we shall call pure state quantum theory is like quantum theory subject to the restriction that the states of the global system are pure. 

Formally, we shall say that a density operator $\rho$ on Hilbert space $\mathcal{H}$ is  a \emph{pure density operator} on $\mathcal{H}$   if there exist a unit vector $ | \tau \rangle$ in $\mathcal{H}$ such that $ \rho = | \tau \rangle \! \langle \tau | $ .

Thus in pure state quantum theory,  the phenomenal space for the global  system $S$ is instead:
\[ \textsf{Pure-Phenomenal-Space}^{S} \isdef \{ \rho^{S} \colon \rho^{S}   \text{ is a  pure density operator on }  \mathcal{H}^{S} \} \, . \] 

For every other system $A$, the phenomenal space of that system shall consist exclusively  of density operators on $\mathcal{H}^{A}$ that are partial traces of pure density operators belonging to $\mathcal{H}^{S}$.
Formally,
\[
\textsf{Pure-Phenomenal-Space}^{A}  \]
\[ \isdef \]
\[  \{ \rho^{A} \colon \exists \rho^{S} \in \textsf{Pure-Phenomenal-Space}^{S} \text{ such that } \rho^{A} = \mathrm{tr}_{\overline{A}} ( \rho^{S}) \} \; .  \]

Everything else remains unchanged.

\subsection{Operations and actions}
We now want to describe how states can evolve in a local-realistic theory.
The evolution of a system happens through operations that can be applied to the system.  The precise operation that is applied  might depend on: the environment external to the `global' system which might be empty if the global system is the whole universe, on the dynamical laws of physics; or on time and other variables.

The next definition  characterizes the mathematical properties of a set of operations that can be performed on a system and how these operations may be combined into to new operations.

\begin{definition}[Monoid of Operations]  A monoid of operations is a 3-tuple $ ( \textsf{Operations}, \circ , I ) $, where $\textsf{Operations}$ is a set whose elements are called \emph{operations}. The set of operations comes with a binary operator denoted~``$\circ$'', called the \emph{composition}, and $I$ is an operation called the \emph{identity operation}.  A monoid of operations satisfies the following properties:
\begin{enumerate}
\item If $U$ and $V$ are operations, $U \circ V$ is an operation called the \emph{composite} of $U$ and $V$;
\item If $U$, $V$ and $W$ are operations, $U \circ  \left( V \circ  W \right) = \left( U \circ V \right) \circ W $;
\item  For all operations $U$,
\[I \circ U = U \circ I = U \, . \]
\end{enumerate}

\noindent
When there is no ambiguity, we shall omit the composition operator and write $UV$ instead of $U \circ V$.

\end{definition}

Now:

\begin{axiom}[Operations on a system]\label{op}
Associated to a system  $A$   is a monoid  of operations, $(\textsf{Operations}^{A} , \circ^{A}, I^{A} ) $.
\end{axiom}

Particular operations on system $A$ are denoted $U^{A}$, $V^{A}$, etc. Also, $I^{A}$ denotes the identity operation on system $A$.  When there is no ambiguity, we drop the superscript and write simply $U$, $V$ and~$I$.

Intuitively, the operations associated to a system are the operations that can be performed on the system, and if $U$ and $V$ are operations on a system, the operation $U \circ V$ can be implemented by first doing $V$, followed by doing $U$.

\paragraph{Operations in Quantum Theory.}
In  quantum theory, the set of  operations on system $A$ with associated Hilbert Space $\mathcal{H}^{A}$ is
\[ \textsf{Operations}^{A} = \{ U \, \colon \, U \text{ is a unitary operator on } \mathcal{H}^{A} \}  \, . \]
The composition of operations $U$ and $V$ is simply their composition as linear operators.  Finally, $I^{A}$ is the identity operator on $\mathcal{H}^{A}$.

An operation can act on a state to produce a new state, as follows:

\begin{definition}[Operation action]
Let  $(  \textsf{Operations}, \circ, I ) $ be a monoid of operations and $S$ be a set.
An~\emph{operation action} of the monoid of operation on set $S$ is a binary operator \mbox{$\star : \textsf{Operations} \times S \to S$} that satisfies, for all operations $U$ and $V$ and for all element $s$ of the set $S$,
\begin{enumerate}
\item $ U \star \left( V \star s \right) = \left( UV \right) \star s $\,;
\item  $I \star s = s$.
\end{enumerate}
\end{definition}

The next two axioms use the above definition to express the fact that an operation done on a system changes its underlying noumenal and phenomenal state.

\begin{axiom}[Noumenal action]\label{axiom:noumenalaction}
Associated to a system $A$, is a \emph{noumenal action}  denoted ``$\star^{A}$'', which is an operation action of the monoid of operations of  the system on the set of noumenal states of the system.
\end{axiom}

When there is no ambiguity, we drop the superscript of the noumenal action.  For example, we write  $U \star N^{A}$ instead of $U \star^{A} N^{A}$.

Intuitively, if  system $A$ was in noumenal state $N^{A}$, and an operation $U$ is done on it, its new noumenal state is $U \star N^{A}$.

\begin{axiom}[Phenomenal action]\label{actionpheno}
Associated to a system $A$, is a \emph{phenomenal action} ``$\cdot^{A}$'', which is an operation action of the monoid of operations of the system on the set of phenomenal states of the system.
\end{axiom}

When there is no ambiguity, we drop the superscript of the phenomenal action. For example, we write $U \cdot \rho^{A}$ instead of $U \cdot^{A}\rho^{A} $.

Intuitively, if a system $A$ was in phenomenal state $\rho^{A}$, and an operation $U$ is done on system $A$, its new phenomenal state is $U \cdot \rho^{A}$. 

\paragraph{Phenomenal action in  quantum theory.} In  quantum theory we do not have yet noumenal states, and hence no noumenal action either, but the phenomenal action of a unitary operation $U$ on a state $\rho$ is defined by the theory  as:
\[ U \cdot \rho \stackrel{\text{def}}{=} U \rho \, U^{\dagger} \, , \]
where $U \rho  \, U^{\dagger} $  is obtained through the usual composition of linear operators $U$, $\rho$ and $U^{\dagger} $ .

\subsection{Noumenal-phenomenal homomorphism}
In a local-realistic theory, what is observable locally in a system is determined by the complete description of that system, in other words, the noumenal state of a system determines its phenomenal state.  If the noumenal state of a system evolves according to an operation, its corresponding phenomenal state must evolve according to the same operation.
 Mathematically the phenomenal state of a system will be determined by its underlying noumenal state, through a structure-preserving surjective map -- a noumenal-phenomenal epimorphism. But first:

\begin{definition}[Noumenal-phenomenal homomorphism] Let $A$ be a system and let $\phi$ be a mapping whose domain is the noumenal state space of $A$ and whose range is the phenomenal state space of $A$.  We~say that $\phi$ is a \emph{noumenal-phenomenal homomorphism} on system $A$ if, for any operation $U$  of system $A$ and any noumenal state $N$ of  $A$,
\[  \phi \! \left(  U \star N \right) = U  \cdot \phi \! \left( N \right ) \, . \]  
\end{definition}

When no ambiguity can arise, we omit  the noumenal action.  For example, the equation above becomes
\[  \phi \! \left(  U   N  \right) = U  \cdot  \phi \! \left(  N \right)  \,  . \]
However,  we never omit writing the phenomenal operation action because there is a possible ambiguity in quantum theory at the phenomenal level:  we need to distinguish between the action of unitary operation $U$ on phenomenal state $\rho$, and the composition of linear operator $U$  with linear operator $\rho$.

\begin{definition}[ Noumenal-phenomenal epimorphism] A surjective noumenal-phenomenal homomorphism on system $A$  is called a \emph{noumenal-phenomenal epimorphism} on system $A$.
\end{definition}

\begin{axiom}[Noumenal-phenomenal epimorphism]
Associated with each system $A$ is  a noumenal-phenomenal epimorphism denoted~$\varphi^{A}$ called \emph{the} noumenal-phenomenal epimorphism of system $A$.
\end{axiom}
When there is no ambiguity, we write $\varphi$ \mbox{instead} of~$\varphi^{A}$. 

Intuitively, if a system $A$ is in noumenal state $N^{A}$, it has a corresponding phenomenal state $\rho^{A} = \varphi ( N^{A} )$. Furthermore, if an operation $U$ is done on system $A$, its new noumenal state is $ U  N^{A}$, and its corresponding new phenomenal state is $U \cdot \rho^{A}$.  Lastly, every phenomenal state on a system arises from at least one noumenal state, since what is observable has an underlying reality.

In unitary quantum theory, we  have not yet  provided noumenal states, nor therefore any noumenal-phenomenal epimorphism.  One task ahead is to supply it.

\subsection{Partial traces}
The next axioms express the fact that the noumenal state of a system determines the noumenal state of any of its subsystems.

\begin{axiom}[Noumenal partial trace]\label{sc:projectors}
Associated to all pairs of disjoint  systems $A$ and $B$, is a function denoted $\mathrm{tr}_{B}^{AB}$, which is called the \emph{ noumenal partial trace} of system $B$ from system $AB$.
Partial trace $\mathrm{tr}_{B}^{AB}$ is a surjective function from the noumenal space of system $AB$ to the noumenal space of system $A$.

Furthermore, for any noumenal state $N^{ABC}$ belonging to a composite system $ABC$, then the following relation {must} hold between partial traces:
\[ \left(  \mathrm{tr}_{B}^{AB} \circ \mathrm{tr}_{C}^{ABC} \right) \left( N^{ABC} \right) = \mathrm{tr}_{BC}^{ABC} \! \left( N^{ABC} \right) . \]
\end{axiom}

When there is no ambiguity, we shall omit the superscript  and we shall refer to $\mathrm{tr}_{B} $ as \emph{the}
noumenal partial trace  of system~$B$, regardless of the supersystem from which it is obtained by tracing. 
For example, the previous equation will simply be written as:
\[ ( \mathrm{tr}_{B} \circ \mathrm{tr}_{C} )( N^{ABC} ) = \mathrm{tr}_{BC} ( N^{ABC} ) \, . \]

Intuitively if a composite system $AB$ is in noumenal state $N^{AB}$, the noumenal state of system $A$ will be $N^{A} = \mathrm{tr}_{B} ( N^{AB} ) $.

The next axiom expresses the fact that the phenomenal state of a system determines the phenomenal state of any of its subsystems.  

\begin{axiom}[Phenomenal partial trace]\label{tracepheno}
Associated to all pairs of disjoint systems $A$ and $B$, is a function  called the \emph{phenomenal partial trace} of system $B$ from system $AB$. These phenomenal partial traces follow the same requirements as noumenal partial traces, {as stated in axiom~\ref{sc:projectors}},
\emph{mutatis mutandis}.  As an abuse of notation, we also \mbox{denote} the phenomenal partial trace of system $B$ from system $AB$ by $\mathrm{tr}_{A}^{AB}$, since no ambiguity will be possible with the corresponding noumenal partial trace $\mathrm{tr}_{A}^{AB}$.
\end{axiom}
Finally, when there is no ambiguity, we shall omit the superscript  and we shall refer to $\mathrm{tr}_{B} $ as \emph{the}
phenomenal partial trace  of system~$B$, regardless of the supersystem from which it is obtained by tracing. 

Intuitively if a composite system $AB$ is in phenomenal state $\rho^{AB}$, the phenomenal state of system $A$ will be $\rho^{A} = \mathrm{tr}_{B} ( \rho^{AB} ) $.

\paragraph{Partial traces in quantum theory.}
In quantum theory, the phenomenal partial traces  are the familiar partial traces of linear algebra.  However, we do not have yet  noumenal partial traces.

The next axiom expresses the consistency between the noumenal and phenomenal partial traces.
\begin{axiom}[Relation between noumenal and phenomenal partial traces] \label{ax:relnouphe}
For all composite systems $AB $, and all noumenal states $N^{AB} $, the noumenal and phenomenal partial traces of  the system $B$ are related by the following commuting relation:
\[
\mathrm{tr}_{B} \! \left( \varphi\! \left( N^{AB} \right) \right) = \varphi \! \left( \mathrm{tr}_{B} \! \left( N^{AB} \right) \right) .
\]
Here the symbol $\varphi$ stands for $\varphi^{AB} $ on the left side of the equation, but for $\varphi^{A} $ on the right side.  Also the symbol $\mathrm{tr}_{B}$ stands for the phenomenal partial trace on the left side of the equation, but for the noumenal partial trace on the right side. 
\end{axiom}

\subsection{Noumenal product}
 We shall need the following definition to express the next axiom.

\begin{definition}[Compatible states]
Let $AB$ be a composite system, let $N^{A}$ and $N^{B}$  be noumenal states of system $A$ and $B$ respectively.  We say that $N^{A} $ and $N^{B}$ are \emph{compatible states} if there exist a noumenal state $N^{AB}$ of system $AB$ such that $N^{A} = \mathrm{tr}_{B} ( N^{AB} ) $ and $ N^{B} = \mathrm{tr}_{A} ( N^{AB} ) $.
\end{definition}

The next axiom  expresses the requirement  that the state of a system is determined by the state of its subsystems.

\begin{axiom}[Noumenal product]\label{sc:join}
Associated to all {disjoint} systems $A$ and $B$ is an operation, the \emph{noumenal product}, denoted  ``$\odot$'',\footnote{\,Technically, we should write $\odot_{( A, B)}$ to denote the fact that the noumenal product depends on systems $A$ and $B$, but since there will be no confusion, as an abuse of notation, we shall not do so.}
\newcounter{fnjoin}%
\setcounter{fnjoin}{\thefootnote}%
such that for all noumenal states $N^{AB} $,  the following equation holds:
\[ N^{AB} =  \mathrm{tr}_{B} \! \left( N^{AB} \right) \odot  \mathrm{tr}_{A} \! \left( N^{AB} \right) \, . \]
Furthermore, $N^{A} \odot N^{B}$ is defined only  if states $N^{A}$ and $N^{B}$ are compatible.
\end{axiom}

There is no corresponding axiom  for a   phenomenal product, and this  is the most profound distinction between the noumenal and phenomenal level.

Unitary  quantum theory does not yet provide us with a noumenal product, so we shall  need to provide it.

The following theorems give important properties of the noumenal product.

\begin{theorem}\label{traNOU} When $N^{A} \odot N^{B}$ is defined,
\[ \mathrm{tr}_{B} \! \left( N^{A} \odot N^{B} \right) = N^{A} \text{ and } \mathrm{tr}_{A} \! \left( N^{A} \odot N^{B} \right)  = N^{B} \, . \]
\end{theorem}
\begin{proof}
We prove only $\mathrm{tr}_{B} \! \left( N^{A} \odot N^{B} \right)  = N^{A} $; the other statement is similar. Since $N^{A} \odot N^{B}$ is defined, there exist $N^{AB}$  such that $\mathrm{tr}_{B} ( N^{AB}) = N^{A} $ and $\mathrm{tr}_{B} ( N^{AB} ) = N^{B} $.   Therefore $N^{A} \odot N^{B} = \mathrm{tr}_{B} ( N^{AB})  \odot \mathrm{tr}_{A} ( N^{AB})  =  N^{AB}$. We conclude that,
\[
\mathrm{tr}_{B} \! \left( N^{A} \odot N^{B} \right) = \mathrm{tr}_{B} ( N^{AB} ) = N^{A} \, . \]
\end{proof}

\begin{theorem}[Unique Decomposition]\label{thm:uniquedec}
Let $A$ and $B$ be disjoint systems, then
\[ N^{A}_{1} \odot N^{B}_{1} = N^{A}_{2} \odot N^{B}_{2} \; \implies  \;  N^{A}_{1} = N^{A}_{2} \text{ and } N^{B}_{1} = N^{B}_{2}  \, . \] 
\end{theorem}
\begin{proof} 
We prove $N^{A}_{1} = N^{B}_{1} $; the other statement is similar. We apply   theorem \ref{traNOU} and get:
\begin{align*}
& ~N^{A}_{1} \\
= &  ~\mathrm{tr}_{B} \left( N^{A}_{1} \odot N^{B}_{1} \right)  \\
= & ~\mathrm{tr}_{B} \left( N^{A}_{2} \odot N^{B}_{2} \right) \\
= & ~N^{A}_{2} \, .  \\[-7ex]
\end{align*}
\end{proof}

\begin{theorem}\label{thm:unicnoume}
For any arbitrary noumenal state $N^{AB}$, there exists a unique $N^{A}$ and a unique $N^{B}$ such that $N^{AB} = N^{A} \odot N^{B}$ .
\end{theorem}
\begin{proof} We first prove existence. By definition of the noumenal product, $N^{A} = \mathrm{tr}_{B} ( N^{AB} ) $ and $N^{B} = \mathrm{tr}_{A} ( N^{AB} ) $ verify $N^{AB} = N^{A} \odot N^{B} $.
Unicity was already obtained in theorem \ref{thm:uniquedec}.
\end{proof}

\subsection{Product of operations}
Intuitively, the next axiom tells us what happens to the noumenal state of  composite system $AB$, if an operation $U$ is performed on system $A$ and an operation $V$ is performed simultaneously on system $B$.

\begin{axiom}[Product of operations]\label{sc:SEPO}
Associated to all disjoint systems $A$ and $B$, is  a \emph{product of operations}, which we \mbox{denote} ``$\times$''.\footnote{ \label{properations} Technically, we should denote this product of operations as $\times_{A, B}$ but we shall consider the dependence on $A$ and $B$ to be implicit. } 

Given  operations $U$ on system $A$ and $V$ on system $B$, $U \times V$, the \emph{product} of $U$ and $V$, is an operation on system $AB$ that satisfies the following relation for any noumenal state $N^{AB} = N^{A} \odot N^{B}$ :
\[
\big( U \times V \big) \! \left( N^{A} \odot N^{B} \right) =  (U N^{A} )  \odot  ( V   N^{B} )  \,  .
\]
\end{axiom}

\paragraph{Product of operation in quantum theory}
In unitary quantum theory, the product of operations is  the usual tensor product $ \otimes$.\footnote{\,Technically, as mentioned earlier,  we should denote this product of operations as $\otimes_{A, B}$ but we shall consider the dependence on $A$ and $B$ to be implicit. } 
\newcounter{fn}%
\setcounter{fn}{\thefootnote} %
 Of course, we will need to verify that the tensor product satisfies axiom \ref{sc:SEPO}.

This concludes our statement of the axioms that characterize a local-realistic theory.

\subsection{No action at a distance versus no-signalling}\label{sc:NSP}
We now state  two theorems that hold for arbitrary local-realistic theories.  This  will help us clarify the  difference between two important facets of locality, which   we shall need   when we analyze Everett's conception of quantum theory and contrast it  with the approach in the present article.

In a local-realistic theory, no action on a system can have any effect whatsover on any other non-interacting system. 
This is formalized as follows:

\begin{theorem}[No-action at a distance]\label{th:NAD}
Let $N^{AB}$ be a noumenal state {of composite system~$AB$}\@. For all operations $U$ on system $A$ and $V$ on system $B$,
\[
\mathrm{tr}_{B} \! \left( \left( U \times V \right)  \left( N^{AB} \right) \right)  = U \,    \mathrm{tr}_{B}  \! \left( N{^{AB}} \right) \,  .
\]
\end{theorem}
\begin{proof}
Let $N^{A} = \mathrm{tr}_{B} ( N^{AB} ) $ and let $N^{B} = \mathrm{tr}_{A} ( N^{AB} ) $.
\begin{align*}
\mathrm{tr}_{B} \! \left( \left( U \times V \right)  \left( N^{AB} \right) \right)  =& \mathrm{tr}_{B}(  ( U \times  V) (  N^{A} \odot N^{B} ) ) \\
                                                                                                                       = & \mathrm{tr}_{B} (  U N^{A} \odot V N^{B}  ) \\
                                                                                                                      = & U   N^{A} \\
                                                                                                                      = & U \,  \mathrm{tr}_{B}  \! \left( N{^{AB}} \right) \,  .
\end{align*}
\end{proof}

 Intuitively the \emph{no-action at a distance theorem} states that no operation $V$ performed on system $B$ can have a noumenal  effect on a remote system~$A$.

This theorem implies  another important, albeit obvious, consequence of a theory being local-realistic.  If no action on a system can have any effect on any other non-interacting system, then no action on a system can have any \emph{observable} effect on any other non-interacting system.  This is formalized in the following theorem, valid in all local-realistic theories, which is a generalisation of a well-known theorem of quantum theory.

\begin{theorem}[No-Signalling Principle]\label{th:NSP}
Let $\rho^{AB}$ be a phenomenal state {of system~$AB$}\@. For all operations $U$ on system $A$ and $V$ on system $B$,
\[
\mathrm{tr}_{B} \! \left( \left( U \times V \right) \cdot  \rho^{AB} \right)  = U \cdot  \mathrm{tr}_{B}  \! \left( \rho{^{AB}} \right) .
\]
\end{theorem}
\begin{proof}
See ref. \cite{STRUCTURE}.  
\end{proof}

Intuitively the \emph{no-signalling principle}  means that no operation $V$ applied on system $B$ can have a phenomenal (i.e.~observable) effect on a remote system~$A$. 

The no-signalling principle is the phenomenal counterpart of the no-action at a distance theorem. It can be satisfied in a purely phenomenal theory, one without any mention of noumenal states, in contrast with the no-action at distance theorem, which refers to  noumenal states.  \footnote{See ref. \cite{STRUCTURE} for more information on the notion of  a  purely phenomenal theory that satisfies the no-signalling principle.}

\section{Interlude - The universal wave function cannot be a complete description of a  local reality}
\label{Incomplete}
In his doctoral dissertation Hugh Everett III introduced a fascinating idea, the theory of the universal wave function which led to the development of unitary quantum theory  \cite{ThesisEverett,Everett}.

The most important ideas of the  theory are:
\begin{itemize}
\item At all time, the state of universe is described, up to irrelevant phase factor,  by a unit vector in the Hilbert space of the universe, called the \emph{universal wave function} of the universe.
\item All systems, including observers, are quantum systems, there are no classical systems.
\item  All interactions between  systems are described by unitary operations. This implies that  what was previously described as a measurement of a quantum system by an observer is reduced to an ordinary unitary operation involving the observing  system, the observed system, and the environment.  Also, any interaction between what was previously described as two classical systems, for example the sending of classical information,  is also described by a unitary operation.
\item The universal wave function of the universe evolves according to a global unitary operation.
\end{itemize}

Importantly, the goal of Everett was not to give a local-realistic explanation  for quantum theory, but rather to remove the arbitrary distinction between classical systems and quantum systems and therefore to analyze  all interactions between systems as  quantum interactions. 
  Nevertheless he believed that his approach could resolve  the various paradoxes that plagued quantum theory, including the Einstein-Podolsky-Rosen paradox \cite{EPR}.   
As he wrote in his  thesis \cite{ThesisEverett}:
\begin{quote}
It is therefore seen that one observer's observation upon one system
of a correlated, but non-interacting pair of systems, has no effect on the
remote system, in the sense that the outcome or expected outcome of any
experiments by another observer on the remote system are not affected.
Paradoxes like that of Einstein-Rosen-Podolsky which are concerned
with such correlated, non-interacting, systems are thus easily understood
in the present scheme.
\end{quote}
Here Everett is expressing in the language of probabilities  the idea that the theory of the universal wave function obeys the no-signalling principle,  i.e. theorem \ref{th:NSP}. \footnote{See also ref.  \cite{STRUCTURE} for  an explanation of the probabilistic version of the no-signalling principle  and how it relates to the formulation given in the present article.}

In his thesis, Everett described the states of subsystems as mixtures of pure states.  This description determines the density operator of a given subsystem.  Nevertheless he recognized that the mixture of pure state  description  was insufficient to give a  fully  local description of entangled subsystems.  As he wrote in his thesis \cite{ThesisEverett}:
\begin{quote}
However, this representation by a mixture must be regarded as only a
mathematical artifice which, although useful in many cases, is an incomplete description because it ignores phase relations between the separate
elements which actually exist, and which become important in any interactions which involve more than just a subsystem.
\end{quote}
Thus Everett was well aware of the problem of locality in quantum theory, and knew that he did not give a complete local description of the states of subsystems,  merely  a useful description of what is locally observable in a subsystem. Describing the state of subsystems as  mixture of pure states  is silent about  the most important  facet of locality, namely the possibility to retrieve the complete state of a system from the states of its subsystems, which is expressed mathematically by the existence of a noumenal product. This also implies that  the no-action at a distance principle, theorem \ref{th:NAD},  cannot be  addressed with this description.

But what about the universal wave function itself,  is it too a ``mathematical artifice'', insufficient to describe a truly local-realistic theory, but nevertheless a useful description?  Or could we take the wave function as a complete local-realistic description of the state of the universe, and find a truly local description of the states of subsystems?

Thus our purpose is  to answer the following question: ``Can the universal wave function be a complete description of a \emph{local} reality?''  Or, to phrase it according to our terminology, is there a local-realistic model of quantum theory where the noumenal state of the universe be described by the universal wave function proposed by Everett?
Certainly the universal wave function determines everything that is observable in the universe.  Therefore, at the very least,  it  determines the phenomenal state of the universe.  

We now prove that  the universal wave function cannot be a local-realistic description of the noumenal state of the universe.
For this purpose, suppose we are  given a local-realistic model of  pure state  quantum theory.

Thus for each system $A$, we are given:
\begin{enumerate}
\item A Hilbert space $\mathcal{H}^{A}$.
\item A noumenal space, which consist of a  set of noumenal states. 
\item A set of phenomenal states, namely $\textsf{Pure-Phenomenal-Space}^{A}$, whose elements are density operators on system $A$.
\item  A monoid of operations which consist of the  unitary operations on $\mathcal{H}^{A}$.
\item An action of the unitary operations on the phenomenal states, which is the usual action of a unitary operation on a density operator.
\item An action of the unitary operations on the noumenal state space.
\item A phenomenal partial trace, which consist of the usual tracing out of the system $A$ from a density operator on a larger system.
\item A  noumenal partial trace.
\item A noumenal-phenomenal epimorphism.
\end{enumerate}

For each disjoint system $A$ and $B$, we are also given a noumenal product and a product of operation, which is the usual tensor product.

We now make the following three suppositions, which  correspond to the idea that the universal wave function  is a description of the noumenal state of the universe. 
\begin{enumerate}
 \item The set of noumenal states of the global system $S$  are  the  unit vectors of  Hilbert space $\mathcal{H}^{S}$.   
\item  The noumenal action on the global system of unitary operation $U$ on a unit vector $| \tau \rangle$ is simply the usual application of the linear operator $U$ on the  vector $ | \tau \rangle $.  Formally,
\[ U \star | \tau \rangle \isdef U ( | \tau \rangle ) \, . \]
\item The noumenal-phenomenal epimorphism $\varphi^{S}$ on the global system $S$ is the  function that sends a unit vector $| \tau \rangle $  to its corresponding density operator $ | \tau \rangle \! \langle \tau | $.  
\end{enumerate}

The reader can easily verify that we have indeed defined a noumenal action on the global system  and that  $\varphi^{S}$ is indeed a noumenal-phenomenal epimorphism on the global system.

Lastly, we need to make the additional hypothesis that the theory contains at least one non-trivial system.  Recall that a system $A$ is non-trivial if it is neither the empty system nor  the global system $S$.  This hypothesis is made so that the theory contains at least two disjoint systems, each of which containing at least one qubit.
Using this hypothesis, we take an arbitrary non trivial system $A$.  System $A$ being non-trivial imply that its complement $B$ is also  non-trivial. Since neither  $A$ nor $B$  is the empty  system,   their associated  Hilbert space $\mathcal{H}^{A}$ and $\mathcal{H}^{B}$ are of dimension greater than one.
Thus the Hilbert space of system $A$ contains at least two orthonormal  states $ |0 \rangle^{A} $ and $| 1 \rangle^{A}$, and the Hilbert space of its complement, the system $B$, contains at least two orthonormal states $ |0 \rangle^{B} $ and $| 1 \rangle^{B}$.

Our argument will make use of the following Bell States:
\[ \ket{\Psi^+}^{AB} =   \frac{1}{\sqrt{2}} \left( \ket{0 1}^{AB} + \ket{10}^{AB} \right) \]
 and 
\[ \ket{ \Phi^{+}}^{AB} =   \frac{1}{\sqrt{2}} \left( \ket{0 0}^{AB} + \ket{11}^{AB} \right) \, . \]

In a local-realistic theory, all noumenal states can be written as a product state, in particular the Bell State,
$ | \Psi^+ \rangle^{AB} $
can be written as a product state:
 \[  |  \Psi^{+} \rangle^{AB}  =  \lbrack \Psi^{+} \rbrack^{A} \odot \lbrack \Psi^{+} \rbrack^{B} \, . \]
 where $ \lbrack \Psi^{+} \rbrack^{A} = \mathrm{tr}_{B} \left( | \Psi^{+} \rangle^{AB} \right) $ and $ \lbrack \Psi^{+} \rbrack^{B} = \mathrm{tr}_{A} \left( | \Psi^{+} \rangle^{AB} \right) $. Here, quite importantly,  $\mathrm{tr}_{A}$ and $\mathrm{tr}_{B}$ are noumenal traces, therefore $ \lbrack \Psi^{+} \rbrack^{A} $ and $ \lbrack \Phi^{+} \rbrack^{B}$ are noumenal states of system $A$ and $B$ respectively and they might not in general be density operators or mixture of pure states.  
 
According to axiom  \ref{sc:SEPO}, if  we apply the Pauli operation, $X = \begin{pmatrix} 0 & 1 \\ 1 & 0 \end{pmatrix} $, on both parts of the state,  we get
\[ \left( X \otimes X \right) | \Psi^{+} \rangle^{AB}  = ( X  \lbrack \Psi^{+} \rbrack^A)  \odot ( X  \lbrack \Psi^{+} \rbrack^B ) \, . \]
But since 
 \[  |  {\Psi^{+}} \rangle^{AB} = \left( X \otimes X \right) | \Psi^{+} \rangle^{AB} \, , \] 
by tracing out $B$ and applying theorem \ref{traNOU} we obtain 
\[   \lbrack  \Psi^{+} \rbrack^A = X   \lbrack  \Psi^{+} \rbrack^A  \, . \]

However by using axiom \ref{sc:SEPO} again, we obtain 
\[  |  \Phi^+ \rangle^{AB} = \left( X \otimes I \, \right) | \Psi^+ \rangle^{AB} = ( X  \lbrack \Psi^{+} \rbrack^A  ) \odot \lbrack \Psi^{+} \rbrack^B = \lbrack \Psi^{+} \rbrack^A \odot \lbrack \Psi^{+} \rbrack^B = | \Psi^+ \rangle^{AB} \, . \]

Since  $| \Phi^+ \rangle^{AB}  \neq  | \Psi^{+} \rangle^{AB} $, we have reached a contradiction!

It follows that the universal wave function cannot be the complete description of a local universe.
It~merely determines what can be observed, namely the phenomenal state of the universe.
Seen this way, if we believe in a local universe,
the answer to the question posed in the title of the EPR 1935 paper is:
``NO,~[the-then-standard] quantum-mechanical description of physical reality CANNOT be considered complete''.
In~other words, the universal wave function is but a \emph{shadow} of the real world,
which is fully described only at its noumenal level!

 In  appendix  \ref{sc:notinjective} we go further and  show that  in an arbitrary local-realistic model of  quantum theory (containing at least one non-trivial system), for  any non-empty system, the noumenal-phenomenal epimorphism associated with that system  is not injective.  Thus the noumenal space of a (non-empty) system must be distinct from (and larger than)  its corresponding  phenomenal space. 
Therefore  the distinction between the noumenal world and phenomenal world is mathematically necessary.  Under the hypothesis that if a mathematical quantity is necessary to describe reality, that quantity corresponds to something that is
real, and that quantum theory is local-realistic and describes reality, then this larger noumenal world is real and cannot described by its wave function alone.

 \section{Construction of a local-realistic model for quantum theory} \label{sc:construction}
We now construct a local-realistic model for  quantum theory.
Specifically, recall that  we are given a finite lattice of systems. For~each system  of those systems, we are given:
\begin{enumerate}
\item  A Hilbert space,
\item  A phenomenal state space consisting of density operators on the Hilbert space,
\item A monoid of operations  consisting of the  unitary operations on the Hilbert space of the system,
\item  A phenomenal action of the operations on the phenomenal state space, which is the usual application of a unitary operation on a density operator,
\item A partial phenomenal trace, which corresponds to the usual  partial trace of a density operator. 
\end{enumerate}

Our~goal is to construct for each system:
\begin{enumerate}
\item  A noumenal state space composed of noumenal states,
\item  An action of the unitary operations  on the noumenal state space,
\item A noumenal partial trace,
\item A noumenal-phenomenal epimorphism.
\end{enumerate}
For each pair of disjoint systems, we also need   a noumenal product. We also need to verify that all axioms of a local-realistic theory are satisfied.

The proof will be done in two steps, first we will construct a model for pure state quantum theory, then we shall use the objects created in that construction   to construct a model for quantum theory.

Our first goal is to associate to each system $A$ a specific orthonormal basis $\mathcal{B}^{A}$   of its  Hilbert space $\mathcal{H}^{A}$ .   
For that purpose, we  introduce some concepts related to orthonormal bases and their product.

\paragraph{Product of orthonormal bases of Hilbert spaces}
Let $\mathcal{B}_{1}$ be an orthonormal basis of Hilbert space $\mathcal{H}_{1}$ and $\mathcal{B}_{2} $ be an orthonormal basis  of Hilbert space $\mathcal{H}_{2} $. We define the tensor product of $\mathcal{B}_{1} $ and $\mathcal{B}_{2} $ to be
\[ \mathcal{B}_{1} \otimes \mathcal{B}_{2} \stackrel{\text{def}}{=} \{ | v \rangle \otimes | w \rangle \colon | v \rangle \in \mathcal{B}_{1} , | w \rangle \in \mathcal{B}_{2}  \} \, . \] 
It is well-known  that $\mathcal{B}_{1} \otimes \mathcal{B}_{2}$ is an orthonormal basis of Hilbert space $\mathcal{H}_{1} \otimes \mathcal{H}_{2} $.

We shall  now extend the notion of product of two orthonormal bases to a finite product of orthonormal bases. 
Let $\mathcal{A} $ be a set of mutually disjoint systems, and for each system $A$ of $\mathcal{A}$, let  $\mathcal{B}^{A}$ be an orthonormal basis associated to system $A$. We define recursively $\bigotimes\limits_{A \in \mathcal{A} } \mathcal{B}^{A} $ in the following way:
\[ \bigotimes\limits_{A \in \mathcal{A} } \mathcal{B}^{A}  \isdef \begin{cases} \mathcal{B}^{0} & \text{ if } \mathcal{A}  \text{ is the empty set, } \\
                (    \bigotimes\limits_{  A \in \mathcal{A} - \{ B \} }  \! \!  \! \! \! \! \! \mathcal{B}^{A} )  \otimes    \mathcal{B}^{B}  & \text{ if } B \text{ is any element of } \mathcal{A} \, .
\end{cases}\]
It can be verified that $ \bigotimes\limits_{A \in \mathcal{A} } \mathcal{B}^{A} $ is an orthornomal basis of the Hilbert space associated with system $\bigsqcup\limits_{A \in \mathcal{A} } A$.

Furthermore, if $ \mathcal{A} $, $\mathcal{A}'$ are finite sets of disjoint systems such that $\mathcal{A}$ and $\mathcal{A}'$ are disjoint sets, i.e. $\mathcal{A}  \cap \mathcal{A}' = \emptyset$, it can verified that:
\[ \bigotimes\limits_{A \in ( \mathcal{A} \cup \mathcal{A}' ) } \mathcal{B}^{A} = \bigotimes\limits_{ A \in \mathcal{A} } \mathcal{B}^{A} \; \otimes \; \bigotimes\limits_{ A \in \mathcal{A}' } \mathcal{B}^{A} \, .  \]

\begin{theorem}\label{basisortho}
There exists a family $( \mathcal{B}^{A} )_{A \in \mathcal{S}} $ indexed by the set of all systems $\mathcal{S}$, such that for every system $A$,  $\mathcal{B}^{A}$ is an orthonormal basis of the Hilbert space associated with $A$. Furthermore, for all composite system $AB$, the following relation holds between the basis of composite system $AB$ and  the bases of  subsystems $A$ and $B$:
\[ \mathcal{B}^{AB} = \mathcal{B}^{A} \otimes \mathcal{B}^{B} \,  . \]
\end{theorem}
\begin{proof}
For each atomic system $A'$, we associate   an arbitrary orthonormal basis $\mathcal{B}^{A'}$.
If $A$ is any non-atomic system with atomic decomposition $\mathcal{A} $, we associate it with the orthonormal basis $ \mathcal{B}^{A} = \bigotimes\limits_{A' \in \mathcal{A} } \mathcal{B}^{A'} $.  Note that trivially, when a system $A$ is atomic, its atomic decomposition is $ \{ A \} $ and $\mathcal{B}^{A}  = \bigotimes\limits_{A' \in \{ A \} } \mathcal{B}^{A'}$.   

Thus, we have obtained a family of orthonormal bases $ ( \mathcal{B}^{A} )_{ A \in \mathcal{S} } $ such that for all system $A$ with atomic decomposition $\mathcal{A}$, the following relation is satisfied:
\[   \mathcal{B}^{A} = \bigotimes\limits_{ A' \in \mathcal{A}  } \mathcal{B}^{A'} \]

We now verify that the family of  orthonormal bases $( \mathcal{B}^{A} )_{A \in \mathcal{S}}$ has the property that  for every composite system $AB$,  we have $\mathcal{B}^{AB} = \mathcal{B}^{A} \otimes \mathcal{B}^{B} $. Let $AB$ be a composite system,    let $ \mathcal{A}$, $\mathcal{A}'$ be the atomic decomposition of $A$ and $B$ respectively.  Thus $ \mathcal{A} \cup \mathcal{A}' $ is the atomic decomposition of $AB$, and it is easy to see that   $ \mathcal{A} $ and $ \mathcal{A}'$ are disjoint sets.
It follows that:
\begin{align*}
 &  \mathcal{B}^{AB}    \\
= &  \bigotimes_{ A' \in \mathcal{A} \cup \mathcal{A }'} \mathcal{B}^{A'} \\
= & \bigotimes_{ A' \in \mathcal{A} } \mathcal{B}^{A'}   \; \otimes \; \bigotimes_{A' \in \mathcal{A}^{'}  } \mathcal{B}^{A'} \\
= & \mathcal{B}^{A} \otimes \mathcal{B}^{B}  \; .
\end{align*}
\end{proof}

We now fix an arbitrary family of bases $(\mathcal{B}^{A} )_{A \in \mathcal{S}} $ satisfying the hypotheses of theorem \ref{basisortho}.  We shall refer to the basis $\mathcal{B}^{A}$ of system $A$, as the \emph{canonical basis} of  system $A$.

 We shall use this fixed family of canonical bases to construct noumenal states.
 To that end we introduce some definitions and concepts.

Given  vector spaces $V$, we shall denote by  $\mathcal{L} ( V) $  the set of linear maps from $V$ to to itself. 
 More formally,
\[ \mathcal{L} ( V )  = \{ L  \;  | \; L \colon V \to V \text{ is a linear map} \} \,  . \]

\begin{definition}[Operator matrix]
Let $A$ be a system with associated  Hilbert space $\mathcal{H}^{A}$, and let $\mathcal{B}^{A} $ be an orthonormal basis of $\mathcal{H}^{A}$.  An application $M^{A} \colon \mathcal{B}^{A} \times \mathcal{B}^{A} \to \mathcal{L} ( \mathcal{H}^{S} ) $, where $\mathcal{H}^{S}$ is the Hilbert space of the global system $S$,  is called an \emph{operator matrix} in basis $\mathcal{B}^{A}$.
\end{definition}

\begin{definition}[Faithfully indexed set] Let $\mathcal{A} = \{ A_{i} \}_{i \in I} $ be a set indexed by $I$.  We say that $\mathcal{A} $ is faithfully indexed, if for all distinct indices $i$ and $j$, their corresponding elements $A_{i}$ and $A_{j}$ are distinct.
\end{definition}

For the purpose of this article, all indexed orthonormal bases will be implicitly assumed to be faithfully indexed.  This condition will be necessary for sums over indices.

\begin{definition}[Entries of an operator matrix]
 Let $A$ be a system with associated  Hilbert space $\mathcal{H}^{A}$, and let $\mathcal{B}^{A}  = \{ | i \rangle \}_{i \in I } $ be an orthonormal basis of $\mathcal{H}^{A}$  indexed by a set $I$.  Let $M^{A}$ be an operator matrix in basis $\mathcal{B}^{A}$. For each pair, $i$ and $j$, of  elements of the index set $I$,   we define the $ij$ \emph{entry} of the operator matrix $M^{A}$ to be 
\[ M^{A}_{ij} \stackrel{\text{def}}{=} M^{A} ( | i \rangle, | j \rangle ) \, . \]
\end{definition}

Importantly, if two operator matrices share the same entries (with respect with the  same indexing  for  $\mathcal{B}^{A}$), it follows that they are equal.  

Note that while the entries of the operator matrix $M^{A}$ are implicitly dependent on the indexing for the basis $\mathcal{B}^{A}$, the operator matrix  itself is not.

\begin{definition}[Evolution Matrix] \label{evoma}
Let $A$ be a system with associated  Hilbert space $\mathcal{H}^{A}$, let $W$ be a unitary operation on the global system $S$, and let $ \mathcal{B}^{A} = \left\{ | i \rangle \right\}_{i \in I} $ be an orthonormal basis of $\mathcal{H}^{A}$.
The \emph{evolution matrix} $  \left[ W \, \right]^{A}_{ \mathcal{B}^{A}} $  of operation $W$  in the  basis $\mathcal{B}^{A}  $ of system $A$   is 
   an operator matrix in basis $\mathcal{B}^{A}$ whose entry $ij$  is the operator
  \[ \left[ W \, \right]^A_{ij} \isdef W^{\dagger}  ( | j \rangle \! \langle i |  \otimes I^{\overline{A}} ) W \,  .  \]
\end{definition}

\begin{definition}[Evolution Space]
Let $A$ be a system. Let $\mathcal{B}$ be an orthonormal basis of the Hilbert space $\mathcal{H}^{A}$ associated with system $A$.  We define the \emph{space of evolution matrices}  for system $A$, in basis $\mathcal{B}$  as
\[ \textsf{Evolution-Space}^{A}_{\mathcal{B} } \isdef \left\{  \left[ W \right]^{A}_{\mathcal{B}} \colon W \in \textsf{Operations}^{S}  \right\}  . \]
\end{definition}

We give additional  properties and theorems  that relate evolution spaces in different bases in appendix \ref{basischange}.

\paragraph{Noumenal states.}
We  are now ready to define \emph{the} noumenal space of system $A$ in the following way:
Let $A$ be a system and let $\mathcal{B}^{A}$ be its canonical basis,
\[ \textsf{Noumenal-Space}^{A} \stackrel{\text{def}}{=} \textsf{Evolution-Space}^{A}_{\mathcal{B}^{A}} \, . \]

When there is no ambiguity, we shall omit the  subscript denoting the basis of an evolution matrix writing $ \lbrack W \rbrack^{A} $ instead of $ \lbrack W \rbrack^{A}_{\mathcal{B}^{A}}$.  Therefore, when $ \lbrack W \rbrack^{A} $ is a noumenal state of system $A$, the omitted basis is always the canonical basis  of system $A$.

Thus an arbitrary noumenal state $N^{A}$ of system $A$  is an evolution matrix in canonical  basis $\mathcal{B}^{A} = \{ |i \rangle \}_{i \in I}$  for which there exists a unitary operation $ W $ belonging to the global system $S$ such that $N^{A} = \lbrack W \rbrack^{A}$. Therefore noumenal state $N^{A}$ is an operator matrix, and  we shall denote its $ij$ entry by $N^{A}_{ij}$ .

\begin{theorem}\label{noumenal}[Fundamental property of noumenal states] For every system  $A$, every unitary operation  $V$ on system, and every unitary operation $W$ on the global system:
 \[ \big[ W \, \big]^A = \big[ \left( I^A \otimes V  \right) W \big]^A \, . \] 
\end{theorem}
\begin{proof}
Recall that if two matrices share the same entries  then they are equal.
\begin{equation*}
\begin{split}
 \big[ W \, \big]^{A}_{ij}  =& W^{\dagger} \Big( \ket{j}  \! \bra{i} \otimes I^{\overline{A}} \, \Big) W  \\
 =& W^{\dagger} \Big( \ket{j}  \! \bra{i} \otimes \left( V^{\dagger}  I^{\overline{A}} V \right) \, \Big) W \\
 = & W^{\dagger} \left( I^{A} \otimes V^{\dagger} \right) \Big( \ket{j}  \! \bra{i} \otimes I^{\overline{A}} \, \Big) \left( I^{A} \otimes V \right) W \\
 = &  \Big( \left( I^{A} \otimes V \right) W \Big)^{\dagger} \Big( \ket{j}  \! \bra{i} \otimes I^{\overline{A}} \, \Big) \Big( \left( I^{A} \otimes V \right) W \Big) \\
 = & \Big[ \left( I^{A} \otimes V \right) W \Big]^A_{ij}  \, . 
\end{split}
\end{equation*}
\end{proof}

\begin{definition}[Action of a unitary operator on a noumenal state]
Let $A$ be a system, let $\mathcal{B}^{A}= \{ | i \rangle \}_{i \in I} $ be the canonical basis of the associated Hilbert space $\mathcal{H}^{A} $.  Let $U$ be a unitary operation on system $A$.  Let $N^{A}$ be a noumenal state of   $A$.  We define the action of $U$ on $N^{A}$ to be the operator matrix $ U  \star^{A} N^{A} $ in basis $\mathcal{B}^{A}$ whose entry $ij$ is defined as
\[  ( U  \star^{A} N^{A}  )_{ij} \stackrel{\text{def}}{=}  \sum\limits_{k,l} U_{ik} N^{A}_{kl} U^{\dagger}_{lj} \, , \]
where $U_{ik}$  is the entry $ik$ of unitary operator $U$ in basis $\mathcal{B}^{A}$ and $U^{\dagger}_{lj}$ is the $lj$ entry of unitary operator $U^{\dagger}$ in basis $\mathcal{B}^{A}$.  
\end{definition}
  Note that $U_{ik}$ , $U^{\dagger}_{lj}$ are scalars while $N^{A}_{kl}$ is a linear operator. Thus the product $U_{ik} N^{A}_{kl} U^{\dagger}_{lk}$ is obtained through the ordinary multiplication of a linear operator by scalars.

When there is no ambiguity we may omit the superscript of the action symbol and write simply $U \star N^{A}$ instead of $U \star^{A} N^{A}$. Further when there is no ambiguity, we may also omit writing the  $\star$ symbol altogether, for example writing $U N^{A}$ instead of $U \star N^{A} $ .

It remains to prove that we defined a proper action at the noumenal level and thus that axiom \ref{axiom:noumenalaction} is  satisfied. This  follows from the next three theorems.

\begin{theorem}[Fundamental property of the noumenal action]\label{actionproperty} For all system $A$, for all unitary operations on  $A$, and all unitary operations $W$ on the global system $S$, the following holds:
\[ U   [ W \, ]^A  = [ ( U \otimes I^{\overline{A}} ) W \, ]^A \, . \] 
\end{theorem}
\begin{proof}
\begin{equation*}
\begin{split}
 \left( U  [ W \, ]^A  \right)_{ij} & = \sum\limits_{m,n} U_{im} \big[W \, \big]_{mn}^A U^{\dagger}_{nj} \\
 & = \sum_{m,n} \bra{i}  U   \ket{m} \Big( W^{\dagger} \; \big( \ket{n} \! \bra{m} \otimes I^{\overline{A}} \big) \;  W  \Big)  \bra{n}  U^{\dagger}  \ket{j} \\
 & =  \sum_{m,n} W^{\dagger} \; \Big(  \big( \ket{n} \bra{n} \; U^{\dagger}  \ket{j}  \bra{i}  U  \; \ket{m} \bra{m} \big) \otimes I^{\overline{A}} \Big) \; W  \\
 & = W^{\dagger} \; \Big(  U^{\dagger} \ket{j} \!  \bra{i} U \; \otimes \; I^{\overline{A}} \Big) \; W \\
 & = W^{\dagger} \; \big(  U^{\dagger} \otimes I^{\overline{A}} \big) \big( \ket{j} \!  \bra{i} \; \otimes \; I^{\overline{A}} \big)  \; \big( U \otimes I^{\overline{A}} \big) W \\
 & = \big( ( U \otimes I^{\overline{A}} )  W \big)^{\dagger}  \big( \ket{j} \!  \bra{i} \otimes I^{\overline{A}} \big) \big( U \otimes I^{\overline{A}} \big) W \\
 & = \big[ \big( U \otimes I^{\overline{A}} \big) W \big]^A_{ij}  \, .
\end{split}
\end{equation*}
\end{proof}

The  theorem \ref{actionproperty}  implies that the action of a unitary operation on a noumenal state is  indeed a noumenal state. It follows that for all system $A$, we are justified in writing that 
\[ \star^{A} \colon \textsf{Operations}^{A} \times \textsf{Noumenal-Space}^{A} \to \textsf{Noumenal-Space}^{A} \, . \]

\begin{theorem}\label{assoaction}
For every system $A$ and  all operations $U$ and $V$ on system $A$, and  for every  unitary operation $W$ on the global system $S$,
\[ \left( V U \right)  [ W ]^{A}  = V \! \left( U  [ W ]^{A}  \right)   \; . \]
\end{theorem}
\begin{proof}
\begin{align*}
 &  ~ ( V U )  [ W]^{A} \\
 = & ~ [ ( ( VU ) \otimes I^{\overline{A}} \, ) W ]^{A} \\
 = & ~ [ ( V \otimes I^{\overline{A}} \, ) ( U \otimes I^{\overline{A}} \, ) W ]^{A} \\
 = & ~ V  [ ( U \otimes I^{\overline{A}} \, ) W ]^{A}  \\
 = & ~ V \! \left( U  [ W ]^{A}  \right)   \; . \\[-7ex]
 \end{align*}
\end{proof}

\begin{theorem}\label{identityaction} For all system $A$, and every   unitary operation $W$ on the global system $S$, 
\[ I^{A}   [ W ]^{A}   = [ W ]^{A}  \; . \]
\end{theorem}
\begin{proof}
\begin{align*}
 &  ~ I^{A}  [ W]^{A}  \\
 = & ~ [ ( I^{A} \otimes I^{\overline{A}} \, ) W ]^{A} \\
  = & ~ \left[ I^{S} \, W \right]^{A} \\
 = & ~  [ W ]^{A}  \; . \\[-7ex]
 \end{align*}
\end{proof}

There is an important lesson hidden in the proof of the three last theorems.  Notice that we have defined an action and verified that it satisfies  the fundamental property of the noumenal action, namely theorem \ref{actionproperty}.  This theorem has been proved by verifying that the entries of both operator matrix $ U \lbrack W \rbrack^{A} $ and $ \lbrack  (U \otimes I^{\overline{A}}) W \rbrack^{A} $ are equal. A proof that examines only entries of matrices   belongs exclusively  to linear algebra.  However, we have then used theorem  \ref{actionproperty} to prove theorems  \ref{assoaction} and \ref{identityaction},  without reference to  entries of matrices.  Such a proof  can be more easily generalized beyond  linear algebra.
Many of the proofs that follow are built  on similar ideas, where generalization is kept in mind. We will prove that a fundamental property is satisfied by certain linear operators by showing an equality by entries, and then use this fundamental property to prove theorems  without any appeal to properties of the entries of the matrices.

There is another key to generalization.  We could have alternatively taken the fundamental property of noumenal actions to be the definition of the noumenal action. This is what was done in ref \cite{STRUCTURE}. Doing this completely eliminates the need to work with entries of matrices.

\begin{definition}[Noumenal partial trace]
Let $AB$ be a composite system, let $\mathcal{B}^{A} = \{ | i \rangle \}_{i \in I} $ be the canonical basis of system $A$, and  $\mathcal{B}^{B} = \{ | k \rangle \}_{k \in K} $ be the canonical basis of system $B$. Recall that the canonical basis of system $AB$ is 
\[ \mathcal{B}^{AB} = \mathcal{B}^{A} \otimes \mathcal{B}^{B} =  \{ |i \rangle \otimes | k \rangle \colon i \in I , k \in K \}=  \{ |i \rangle \otimes |k \rangle \}_{ (i,k) \in I \times K} \, . \]
Let $N^{AB}$ be a noumenal state of system $AB$.
 Define $\mathrm{tr}_{B}^{AB} ( N^{AB} ) $,  the \emph{noumenal partial trace} of  $N^{AB}$ with respect to system $B$, as the operator matrix in basis $\mathcal{B}^{A} $  whose entry $ij$ is defined as:
\[  ( \mathrm{tr}_{B}^{AB} ( N^{AB} ) )_{ij} \stackrel{\text{def}}{=}  \sum\limits_{k \in K}  N^{AB}_{( i, k ) ( j, k ) } \, . \]
\end{definition}

As usual, when there is no ambiguity, we omit the superscript and write $ \mathrm{tr}_{B} $ instead of $\mathrm{tr}_{B}^{AB}$.

It remains to verify that what we have defined satisfies the axiom  \ref{sc:projectors}  for noumenal partial traces.   This is a consequence of  the  next  three theorems.

\begin{theorem}[Fundamental property of the noumenal partial trace]\label{th:noumparttrace}
Let $AB$ be a composite system, let $\lbrack W \rbrack^{AB} $ be an arbitrary noumenal state of system $AB$,
  \[ \mathrm{tr}_B ( [ W \, ]^{AB} ) = [ W \, ]^{A} \, .  \]
\end{theorem}
 
\begin{proof}
Recall  that  two matrices of operators are the same if they share the same entries.  Let $C$ be the complement of composite system $AB$; thus $ABC = S$, where $S$ is the global system.
\begin{equation*}
\begin{split} 
\left( \mathrm{tr}_B ( [ W \, ]^{AB} ) \right)_{ij} & = 
 \sum_k [ W \, ]^{A B}_{ \left( i,k \right) \left( j,k \right)}   \\
 & = \sum_k W^{\dagger} \left( \ket{ j} \! \bra{i}^{A} \otimes \ket{ k} \! \bra{k}^{B}\otimes I^{C} \right) W \\
 & = W^{\dagger} \left( \ket{j} \! \bra{i}^A \otimes I^{B} \otimes I^{C} \right) W \\
& = W^{\dagger} ( \ket{j} \! \bra{i}^{A} \otimes I^{BC} ) W \\
 & = [ W \, ]^A_{ij} \, .
 \end{split}
 \end{equation*}
   \end{proof}

Theorem \ref{th:noumparttrace} implies  that a noumenal state of system $AB$ is sent to a noumenal state of system $A$ by noumenal  partial trace $\mathrm{tr}_{B}^{AB} $.  Thus we are justified in writing:
\[ \mathrm{tr}_{B}^{AB} \colon \textsf{Noumenal-Space}^{AB} \to \textsf{Noumenal-Space}^{A} \, . \]

\begin{theorem}\label{th:surjpartialtrace}
Noumenal partial trace $\mathrm{tr}_{B}^{AB}  \colon \textsf{Noumenal-Space}^{AB} \to \textsf{Noumenal-Space}^{A}$ is surjective.
\end{theorem}
\begin{proof} An arbitrary noumenal state of system $A$  is equal to  $\lbrack W \rbrack^{A} $ for some operation $W$ on the global system $S$. Since  $ \lbrack W \rbrack^{AB}$ is a noumenal state of system $AB$ and $ \mathrm{tr}_{B}^{AB} ( \lbrack W \rbrack^{AB} ) = \lbrack W \rbrack^{A} $, it follows that $\mathrm{tr}_{B}^{AB}$ is surjective.
\end{proof}

\begin{theorem}\label{th:otherpartialtrace}
Let $ABC$ be a composite system, let $W$ be a unitary operation on global system $S$.
\[  \mathrm{tr}_{BC} \!\left( \left[ W \right]^{ABC} \right) = \! \left(\mathrm{tr}_{B} \circ \mathrm{tr}_{C} \right) \!\left( \left[ W \right]^{ABC} \right) \, . \]
\end{theorem}
\begin{proof}
 \begin{align*} \mathrm{tr}_{BC} \!\left( \! \left[ W \right]^{ABC} \right) = & ~ \! \left[ W \right]^{A} \\
 = &  ~  \mathrm{tr}_{B} \!\left( \left[ W \right]^{AB} \right) \\
 = &  ~ \mathrm{tr}_{B} \! \left( \mathrm{tr}_{C} \! \left[ W \right]^{ABC} \right) \\
 = & ~ \! \left( \mathrm{tr}_{B} \circ \mathrm{tr}_{C} \right) \left( \left[ W \right]^{ABC} \right)  \, . 
 \end{align*}
\end{proof}

\begin{definition}[Noumenal Product]
Let $AB$ be a composite system,  let $\mathcal{B}^{A} = \{ | i \rangle \}_{i \in I} $, and   $\mathcal{B}^{B} = \{ | k \rangle \}_{k \in K}$ be the canonical bases of systems $A$ and $B$ respectively. 
Let $N^{A}$ and $N^{B}$ be mutually compatible noumenal states of systems $A$ and $B$ respectively.  
We define $  N^{A} \odot N^{B}  $,  the \emph{noumenal product} of  $N^{A}$ and $N^{B}$, to be the operator matrix in canonical  basis 
\[ \mathcal{B}^{AB}  = \mathcal{B}^{A} \otimes \mathcal{B}^{B} =  \{ |i \rangle \otimes | k \rangle \colon i \in I , k \in K \}=  \{ |i \rangle \otimes |k \rangle \}_{ (i,k) \in I \times K} \, , \] whose entry $(i ,k ) (j ,l) $ is:
\[  ( N^{A} \odot N^{B} )_{(i,k) ( j,l) }  \stackrel{\text{def}}{=} N^{A}_{ij}  N^{B}_{k l} \, . \]
\end{definition}
Note that $N^{A}_{ij} N^{B}_{kl} $ denotes the ordinary composition of   operators $N^{A}_{ij} $ and $N^{B}_{kl} $ both of which are linear operators belonging to  $\mathcal{L} ( \mathcal{H}^{S} )$  .

 \begin{theorem}[Fundamental property of the noumenal product]\label{th:fundjoin}  For any composite system $AB$, for all unitary operation $W$ on global system $S$, 
  \[\Big[ W \, \Big]^A \odot \Big[ W \, \Big]^B \, = \,  \Big[ W \, \Big]^{AB} \, .   \]
  \end{theorem}
\begin{proof}
Let $C$ be the complement of composite system $AB$, so  $ABC$ is the global system $S$. 
 \begin{equation*}
\begin{split} 
 & \bigg(   \Big[   W \, \Big]^A \odot \Big[ W \, \Big]^B \bigg)_{ \left( i,k \right) \left( j,\ell \right) }   \\
=& \Big[   W \, \Big]^{A}_{ij} \Big[ W \, \Big]^{B}_{k \ell}  \\
 = & \bigg(  W^{\dagger} \left( \ket{j} \! \bra{i}^A \otimes I^{B} \otimes I^{C} \right) W \bigg)
   \bigg(  W^{\dagger} \left( I^{A} \otimes  \ket{\ell}  \! \bra{k}^{B} \otimes I^{C} \right) W \bigg) \\
  = &  W^{\dagger} \left( \ket{j}  \! \bra{i}^A \otimes I^{B} \otimes I^{C} \right)    \left( I^{A} \otimes 
 \ket{\ell} \!  \bra{k}^{B} \otimes I^{C} \right) W \\
  = &  W^{\dagger} \left( \ket{j}  \! \bra{i}^A \otimes  \ket{\ell}  \! \bra{k}^B \otimes I^{C} \right) W \\
  = &  \Big[ W \, \Big]^{AB}_{(i,k)(j,\ell)}   \, . 
 \end{split}
 \end{equation*}   
\end{proof}

The noumenal product satisfies  axiom \ref{sc:join}  as proved  by the next theorem.

\begin{theorem}\label{th:joinworks}
Let $AB$ a composite system, and let $ \lbrack W \rbrack^{AB} $ be an arbirary noumenal state of system $AB$. It follows that
\[ \mathrm{tr}_{B} ( \lbrack W \rbrack^{AB} ) \odot \mathrm{tr}_{A}  (  \lbrack W \rbrack^{AB} )  = \lbrack W \rbrack^{AB} \, .  \]
\end{theorem}
\begin{proof}
\[ \mathrm{tr}_{B} ( \lbrack W \rbrack^{AB} ) \odot \mathrm{tr}_{A}  (  \lbrack W \rbrack^{AB} ) = \lbrack W \rbrack^{A} \odot \lbrack W \rbrack^{B} = \lbrack W \rbrack^{AB} \, .  \]
\end{proof}

With the next theorem we prove that the product of operations  satisfies axiom \ref{sc:SEPO}.
\begin{theorem}\label{th:UVW}
Let $AB$ be a composite system, let $U$ and $V$ be unitary operations on $A$ and  $B$ respectively and  let $ \lbrack W \rbrack^{AB} = \lbrack W \rbrack^{A} \, \odot \: \lbrack W \rbrack^{B} $ be  an arbitrary noumenal state of system $AB$, the following equation is satisfied: 
\[ \left( U \otimes V \right)  \left( \left[ W \right]^{A} \odot \, \left[ W \right]^{B} \right) = U \left[ W \right]^{A}  \, \odot \; V  \left[ W \right]^{B} \, . \]
\end{theorem}
\begin{proof}
Let $C$ be the complement of composite system $AB$, thus $ABC = S$, where $S$ is the global system.
\begin{align*} & \left( U \otimes V \right) \left( \left[ W \right]^{A} \odot \, \left[ W \right]^{B} \right) \\
= & \left( U \otimes V \right) \left[ W \right]^{AB}  \\
=& \left[ \left( U \otimes V \otimes I^{C} \, \right) W \right]^{AB} \\
=& \left[ \left( U \otimes V \otimes I^{C} \, \right) W \right]^{A} \, \odot \; \left[ \left( U \otimes V \otimes I^{C} \,\right) W \right]^{B} \\
=& \left[ \left( U \otimes I^{B} \otimes I^{C} \, \right) W \right]^{A} \, \odot \; \left[ \left( I^{A} \otimes V \otimes I^{C} \,\right) W \right]^{B} \\
=& ~ U  \left[ W \right]^{A} \; \odot \; V  \left[ W \right]^{B}  \, .
\end{align*}
\end{proof}

\begin{definition}[Noumenal-phenomenal homomorphism]                               
For each phenomenal state $\rho$ of the global system $S$,
we define the application of  homomorphism $\phi_{\rho}^{A}$ on noumenal state $N^{A}$ of system $A$ to be the density  operator $ \phi_{\rho}^{A} ( N^{A} ) $ for which the $ij$ entry    of the corresponding density matrix in canonical basis $\mathcal{B}^{A}$ satisfies:
\[  \left( \phi_{\rho}^{A} (  N^{A}  ) \right)_{ij}   \isdef  \mathrm{tr} \left( N_{ij}^{A}  \rho  \right) \, . \]
\end{definition}

We now  verify that in quantum theory, for each $\rho$ in the phenomenal space of the global system $S$, function $\phi_{\rho}^{A}$ is indeed a noumenal-phenomenal homomorphism on system $A$.   This follows  from the two next theorems. 

\begin{theorem} [Fundamental property of the noumenal-phenomenal homomorphism] \label{th:fundnouphehomo}  For every system $A$, every phenomenal state $\rho$ of the global system $S$ and every unitary operation $W$ of the global system, 
\[ \phi_{\rho}^{A} ( \left[ W \right]^{A})    = \mathrm{tr}_{\overline{A}} \left( W  \cdot \rho  \right) \, .  \]
\end{theorem}
\begin{proof}
\begin{equation*}
\begin{split} 
 \Big( \phi_{\rho}^{A} \left( \left[ W \, \right]^{A} \right)  \Big)_{ij} & = \mathrm{tr} \left(  \big[ W \, \big]^{A}_{ij} \rho \right) \\
 &= \mathrm{tr} \left( W^{\dagger} \left( \ket{j} \! \bra{i} \otimes I^{\overline{A}} \right) W \rho \right) \\
 &= \mathrm{tr} \left( W^{\dagger} \left( \ket{j}  \! \bra{i} \otimes \sum\limits_{k} \ket{k} \!  \bra{k} \right) W \rho \right) \\
  & = \mathrm{tr} \left( \sum\limits_{k} \left( \ket{j} \! \bra{i} \otimes \ket{k} \! \bra{k} \right) W \rho W^{\dagger} \right) \\
  &=  \sum\limits_{k} \left( \bra{i} \otimes \bra{k} \right) W \rho W^{\dagger} \left( \ket{j} \otimes \ket{k} \right) \\
  &=  \left( \mathrm{tr}_{\overline{A}} \left( W \rho W^{\dagger} \right) \right)_{ij}     \\
  &=  \left( \;  \mathrm{tr}_{\overline{A}}  \left( W \cdot \rho  \right) \; \right)_{ij} \, .
\end{split}
\end{equation*}
\end{proof}

Theorem \ref{th:fundnouphehomo} shows that for all systems $A$, the application of $\phi_{\rho}^{A}$ on an arbitrary noumenal state $\lbrack W \rbrack^{A} $ indeed gives a density operator on system $A$, namely the partial trace applied to  the complement of system $A$ on the density operator $ W \cdot \rho = W \rho W^{\dagger} $. Therefore we may write:  
\[ \phi_{\rho}^{A} \colon \textsf{Noumenal-Space}^{A} \to \textsf{Phenomenal-Space}^{A}\, . \]

\begin{theorem}\label{th:ishomo} For each $\rho$ in the phenomenal space of the global system, the function $\phi_{\rho}^{A}$ is a noumenal-phenomenal homomorphism and satisfies
\[ U \cdot   \phi_{\rho}^{A} \! \left( \left[ W \right]^{A} \right)  = \phi_{\rho}^{A} \! \left(  U \left[ W \right]^{A} \right) \]
for any system $A$,   noumenal state $\left[ W \right]^{A}$ and unitary operation $U$ on $A$.
  \end{theorem}
\begin{proof} 
\begin{align*}
\begin{split}
   & U \cdot ( \phi_{\rho}^{A} (  [ W ]^A  ) ) \\
 = & U  \cdot  ( \mathrm{tr}_{\overline{A}} ( W \cdot \rho  ) )  \\
 = & \mathrm{tr}_{\overline{A}} ( \; ( (  U \otimes I^{\overline{A}} ) W ) \cdot \rho  \;  ) \\
 = & \phi_{\rho}^{A}   ( \,  [ (  U \otimes I^{\overline{A}} ) W ]^{A} \,  ) \\
 = &  \phi_{\rho}^{A} \big( U \big[ W \big]^A \big)  \, .
 \end{split}
\end{align*}
\end{proof}

When  there is no ambiguity,  we shall   omit writing the superscript indicating the system on which a noumenal-phenomenal homorphism act.  Thus we shall write $\phi_{\rho}$ instead of $\phi_{\rho}^{A}$.

The next theorem shows that for each $\rho$ in the phenomenal space of the global system $S$, the  noumenal-phenomenal homomorphism $\phi_{\rho}$ satisfies  axiom    \ref{ax:relnouphe} on the relation between noumenal and phenomenal partial traces.
\begin{theorem}\label{th:relnouphetrace}  Let $\rho$  be a phenomenal state belonging to  the global system $S$,
\[ \mathrm{tr}_{B} \! \left( \phi_{\rho} \! \left(  \left[  W \right]^{AB} \right) \right) =  \phi_{\rho} \! \left(  \mathrm{tr}_{B} \! \left( \left[ W \right]^{AB} \right) \right)
\]
for any composite  system $AB$ and   noumenal state $\left[ W \right]^{AB}$.
\end{theorem}
\begin{proof}
Let $C = \overline{AB}$ i.e. $S = ABC$,
\begin{align*}
  &  \mathrm{tr}_B ( \phi_{\rho}  (  \,  [ W \, ]^{AB} \, )  ) \\
= &  \mathrm{tr}_B \left( \mathrm{tr}_{\overline{AB}} \left( W  \cdot \rho \right) \right) \\
= & \mathrm{tr}_{B} \left( \mathrm{tr}_{C} \left( W \cdot \rho  \right) \right) \\
= & \left( \mathrm{tr}_{B} \circ \mathrm{tr}_{C} \right) \left( W \cdot \rho  \right) \\
= & \mathrm{tr}_{BC} ( W \cdot \rho  ) \\
 = & \mathrm{tr}_{\overline{A}} ( W \cdot \rho  ) \\
 = & \phi_{\rho}  (  \,  [ W \, ]^A  \, )  \\
 = &   \phi_{\rho}  ( \mathrm{tr}_B (  \, [ W \, ]^{AB}  \, )) \, .
\end{align*}
  \end{proof}

We now prove that for any pure density operator $\rho$ on the Hilbert space $\mathcal{H}^{S}$ of the global system $S$, function $\phi_{\rho}^{A}$ is a noumenal-phenomenal epimorphism in pure state quantum theory.  This is the purpose of the following two theorems.

\begin{theorem} \label{th:puredefine} Let $\rho$ be a pure density operator belonging to the phenomenal space of the global system $S$,  for each noumenal state $\lbrack W \rbrack^{A}$ of system $A$, $\phi_{\rho}^{A} ( \lbrack W \rbrack^{A} )$ is an element of $\textsf{Pure-Phenomenal-Space}^{A}$.
\end{theorem}
\begin{proof}
Since $\rho $ is a pure density operator of the global system $S$, it follows that $W \cdot \rho$ is also a pure density operator of the global system, and thus that $\mathrm{tr}_{\overline{A}} ( W \cdot \rho ) $ is an element of $\textsf{Pure-Phenomenal-Space}^{A}$.  Since $ \phi_{\rho}^{A} ( \lbrack W \rbrack^{A} ) = \mathrm{tr}_{\overline{A}} ( W \cdot \rho ) $, the theorem follows. 
\end{proof}
By theorem \ref{th:puredefine}, when $\rho$ is a pure density operator it follows that  either:
\[\phi_{\rho}^{A} \colon \textsf{Noumenal-Space}^{A} \to \textsf{Pure-Phenomenal-Space}^{A} \]
or
\[\phi_{\rho}^{A} \colon \textsf{Noumenal-Space}^{A} \to \textsf{Phenomenal-Space}^{A} \]
depending on  whether we work in pure state quantum theory or in quantum theory.

\begin{theorem}  Let $\rho$ be a pure density operator belonging to the phenomenal space of the global system $S$. For each system $A$, the function  
\[\phi_{\rho}^{A} \colon \textsf{Noumenal-Space}^{A} \to \textsf{Pure-Phenomenal-Space}^{A} \] 
is a noumenal-phenomenal-epimorphism in pure state quantum theory.
\end{theorem}
\begin{proof}   Theorem \ref{th:ishomo} implies that  $\phi_{\rho}^{A}$ is a noumenal-phenomenal homomorphism on system $A$, it remains to prove that  
\[ \phi_{\rho}^{A} \colon \textsf{Noumenal-Space}^{A} \to \textsf{Pure-Phenomenal-Space}^{A} \]
 is surjective.
Let $\rho^{A}$ be an element of  $\textsf{Pure-Phenomenal-Space}^{A}$, by definition there exist $\rho^{S}$ a pure density operator of the global system such that $\mathrm{tr}_{\overline{A}} ( \rho^{S} )= \rho^{A} $.  Since $\rho$ and $\rho^{S}$ are pure density operators of the global system, there exist a unitary operation $W$ such that $W \cdot \rho = \rho^{S}$.  It follows that $ \phi_{\rho}^{A} ( \lbrack W \rbrack^{A} ) = \mathrm{tr}_{\overline{A}} ( W \cdot \rho ) = \mathrm{tr}_{\overline{A}} (  \rho^{S} ) = \rho^{A}$.  Therefore 
\[\phi_{\rho}^{A} \colon \textsf{Noumenal-Space}^{A} \to \textsf{Pure-Phenomenal-Space}^{A} \] is surjective.
\end{proof}

Thus, we may fix an arbitrary pure density operator $\rho$  in the phenomenal space of the global system $S$ and associate to each system $A$ the noumenal-phenomenal epimorphism $\phi^{A}_{\rho}$. 
This  complete the construction of a  local realistic model for pure state quantum theory.

Note  that for any pure density operator $\rho$ on global system $S$,   $\phi_{\rho}^{S}$ is a noumenal-phenomenal epimorphism in \emph{pure state} quantum theory, but it is not one in general  in quantum theory, because
  \[ \phi_{\rho}^{S} \colon \textsf{Noumenal-Space}^{S} \to \textsf{Pure-Phenomenal-Space}^{S} \] is surjective,  implying directly that  \[ \phi_{\rho}^{S} \colon \textsf{Noumenal-Space}^{S} \to \textsf{Phenomenal-Space}^{S} \] is not surjective in any theory containing mixed states. There  will be mixed states  in quantum theory whenever the global system $S$ is distinct from the empty system.

Now that we have a model for pure state quantum theory, we will create a model for quantum theory. 
We shall achieve this by the redefinition of the noumenal states, the noumenal actions, the noumenal partial traces, the noumenal products and the construction  of the noumenal-phenomenal-epimorphisms. These new objects will be constructed by using the objects previously constructed for the model in pure state quantum theory.

Let $A$ a system, we define the new noumenal state of system $A$ to be : 
\[ \textsf{New-Noumenal-Space}^{A} \]
\[ \isdef \]
\[ \left\{ \left( N^{A} , \rho \right) \colon N^{A} \in \textsf{Noumenal-Space}^{A} , \rho \in \text{Phenomenal-Space}^{S} \right\} \, .  \]

Thus an arbitrary new noumenal state of system $A$ will be of the form $ ( \lbrack  W \rbrack^{A} , \rho ) $ for some unitary operation $W$ on the global system, and some density operator $\rho$ of the global system.

Consider any two systems $A$ and~$B$.
\[ \left\{ \begin{array}{lr}
\mathrm{tr}'_{B} \! \left( N^{AB}, \rho \right) \isdef \left( \mathrm{tr}_{B} \! \left( N^{AB} \right) , \rho \right)  & \mathrm{~provided~} AB  \text{ is a composite system}  \\[1ex]
\left( N^{A} , \rho \right) \odot' \left( N^{B} , \rho \right) \isdef \left( N^{A} \odot N^{B} , \rho \right) 
& ~~~~\mathrm{~provided~} N^{A} \odot N^{B} \mathrm{~is~defined} \\[1ex]
U  \! \left( N^{A}, \rho \right) \isdef \left( U  N^{A} , \rho \right)
\end{array}
\right. \]

The new noumenal product $\odot'$ is a partial function, where
\mbox{$\left( N^{A} , \rho_{1} \right) \odot' \left( N^{B} , \rho_{2} \right)$} is defined only whenever
\mbox{$\rho_{1} = \rho_{2}$} 
 and that  $A$ and $B$ are disjoint system and that $N^{A} \odot N^{B}$ is already defined according to the original noumenal product~$\odot$.
The new noumenal-phenomenal epimorphism $\varphi^{A}$ on system $A$ is defined as follows:
\[ \varphi^{A} \! \left( N^{A} , \rho  \right) \isdef \phi_{\rho}^{A} \! \left(  N^{A} \right) . \]

It is easy to verify that the new objects constructed, including the noumenal-phenomenal epimorphisms, satisfy all the axioms of a local-realistic theory.   The verification is tedious but straightforward.
The next two theorems are given as way of examples of  how the axioms are verified. 

First,  axiom~\ref{sc:join}.
\begin{theorem}
\[\mathrm{tr}'_{B} \! \left( N^{AB} , \rho \right) \odot' \mathrm{tr}'_{A} \! \left( N^{AB} , \rho \right) = \left( N^{AB} , \rho \right) \, .\]
\end{theorem}
\begin{proof}
\begin{align*}
\mathrm{tr}'_{B} \! \left( N^{AB} , \rho  \right) \odot' \mathrm{tr}'_{A} \! \left( N^{AB} , \rho \right) &= \left( \mathrm{tr}_{B} \! \left( N^{AB} \right) , \rho \right) \odot' \left( \mathrm{tr}_{A} \! \left( N^{AB} \right), \rho \right) \\
&=  \left( \mathrm{tr}_{B} \! \left( N^{AB} \right) \odot \mathrm{tr}_{A} \! \left( N^{AB} \right), \rho \right) \\
&= \left( N^{AB} , \rho \right) \, .  \\[-7ex]
\end{align*}
\end{proof}

We now prove that the axiom   \ref{ax:relnouphe} concerning  the relation between noumenal and phenomenal partial traces is satisfied.
\begin{theorem}
\[
\mathrm{tr}_{B} \! \left( \varphi^{AB} \! \left( N^{AB} , \rho \right) \right) = \varphi^{A} \! \left( \mathrm{tr}'_{B} \! \left( N^{AB} , \rho \right) \right) \, .
\]
Note that on the left, we have the phenomenal partial trace, which is unchanged, while on the right, we have the new noumenal partial trace.
\end{theorem}
\begin{proof}
\begin{align*}
\mathrm{tr}_{B} \! \left( \varphi^{AB} \! \left( N^{AB} , \rho \right) \right)  &= \mathrm{tr}_{B}  ( \phi^{AB}_{\rho} ( N^{AB} ))  \\
                                                                                                                         &= \phi^{A}_{\rho} ( \mathrm{tr}_{B} ( N^{AB} ) ) \\
                                                                                                                        &= \varphi^{A} ( \mathrm{tr}_{B} ( N^{AB} ), \rho )  \\
                                                                                                                       &= \varphi^{A} \! \left( \mathrm{tr}'_{B} \! \left( N^{AB} , \rho \right) \right) \, .
\end{align*}
Note that we used theorem \ref{th:relnouphetrace}.
\end{proof}

Furthermore for each system, $\varphi^{A}$ is indeed a  noumenal-phenomenal epimorphism,
which was the object of the exercise of defining `new' objects.  
To~prove this, it suffices to show that for each system $A$,
$\varphi^{A}$ is a  noumenal-phenomenal \emph{homo}morphism (which is obvious) and that
$\varphi^{A}$ is surjective,
as follows:

\begin{theorem}
Let $A$ be a system,  \[ \varphi^{A} \colon \textsf{New-Noumenal-Space}^{A} \to \textsf{Phenomenal-Space}^{A}\] is surjective.
\end{theorem}
\begin{proof}
Let  $\rho^{A}$ be an arbitrary density operator of system $A$. We show that there exists $(N^{A}, \rho )$ in 
$\textsf{New-Noumenal-Space}^{A}$ such that $\varphi^{A} \! \left( N^{A}, \rho \right)  = \rho^{A}$.
Let $\rho^{\overline{A} } $ be any density operator on system $\overline{A}$, thus $\rho^{A} \otimes \rho^{\overline{A}}$ is a density operator on the global system $S$.  Let $\rho = \rho^{A} \otimes \rho^{\overline{A}}$.  Recall that  $I^{S}$ is the identity operation on global system $S$.
\begin{align*} \varphi^{A} ( \lbrack I^{S}  \rbrack^{A} ,  \rho  ) =  \, & \phi_{\rho}^{A}  ( \lbrack I^{S} \rbrack^{A} ) \\
 = \,  & \mathrm{tr}_{\overline{A}} ( I^{S} \cdot \rho ) \\
                                                                                                               = \,  & \mathrm{tr}_{\overline{A}} (  \rho  ) \\
                                                                                                               =  \, & \mathrm{tr}_{\overline{A}}   ( \rho^{A} \otimes \rho^{\overline{A}}  ) \\
 = \,  &  \rho^{A}  \, .
\end{align*} 
\end{proof}

The existence of
the required noumenal-phenomenal epimorphism is established for each system,
which concludes our construction of a local-realistic model for quantum theory.

\section{Conclusion} \label{sc:conclusion}
The goal of this article was to give a precise definition of local realism and to show that quantum theory is local-realistic by proving that it has a local-realistic model. The key idea allowing this was to distinguish between what can be observed locally in a system and the complete local description of the system.

Prima facie, one might have hoped that, as in classical physics, the noumenal world-as-it-is would be equivalent to the phenomenal world as it could be observed.  But in quantum theory, there exist entangled systems, in states that defy any local-realistic explanation equating observable with real. This forces us to distinguish noumenal from phenomenal.

Given a precise definition of local-realism and a rigorous local-realistic model for quantum theory, we might wonder why such a model exists.
The precise answer is in ref. \cite{STRUCTURE}.  It is that any theory that conforms to the the no-signalling principle has a local-realistic model.  In other words, a theory forbids action at a distance if and only if it forbids \emph{observable} action at a distance.  

That, in turn, leads to perhaps the deepest question: ``Is quantum theory local-realistic?''

In a discussion with the author, David Deutsch remarked: 
\begin{quote}
Locality is not a symmetric property.  If a theory has no local-realistic model, it is obviously nonlocal.  However, if a theory has a local-realistic model, it will also admit nonlocal models. The meaningful way to describe a local theory is to say that a theory is local when it admits a local model and conversely, a theory is nonlocal when it admits no local model.
\end{quote}

Thus we can conclude that  quantum theory \emph{is}  local-realistic since it has a local-realistic model.

\appendix

\section{The noumenal-phenomenal epimorphism cannot be injective} \label{sc:notinjective}
Here our goal is to show that in quantum theory, under the supposition that the theory contains a non-trivial system, the noumenal-phenomenal epimorphism on any non-empty system is not injective.

In some sense, this result is not very surprising, since we saw in section \ref{qtnotlr} that there exist no local-realistic model of quantum theory, where for all systems the density operator of a system is its complete description.   However, this argument via density operators was formulated casually, and without recourse to the formal framework of local-realistic theories.  Once, we had completed the description of the axioms of local-realism, the argument   could have been  rephrased in the following way:
\begin{quote}
In any local-realistic model of quantum theory, containing a non-trivial system,  it is not the case that for all  system $A$, its associated  noumenal-phenomenal epimorphism $\varphi^{A}$   is injective.
\end{quote}
This should be contrasted with the stronger result in this appendix, namely:
\begin{quote}
In any local-realistic model of quantum theory, containing a non-trivial system, it is not the case that there exist a non-empty system $A$ such that its associated  noumenal-phenomenal epimorphism $\varphi^{A}$ is injective.
\end{quote}

For the purpose of this proof, assume we are given a local-realistic model of quantum theory containing a  non-trivial system distinct from the empty system and the global system.  Let $A$ be  any non-trivial system,  and let system $B$ be the complement of system $A$.   System $B$ is distinct from the empty system.
Thus, both systems $A$ and $B$ are non-empty systems whose associated Hilbert spaces $\mathcal{H}^{A}$ and $\mathcal{H}^{B}$  are of dimension greater than one.
Therefore the Hilbert space of system $A$ contains at least two orthonormal  states $ |0 \rangle^{A} $ and $| 1 \rangle^{A}$, and the Hilbert space of its complement the system $B$ contains at least two orthonormal states $ |0 \rangle^{B} $ and $| 1 \rangle^{B}$.  

To prove our result we shall use the density operators corresponding to  the Bell states
 \[ \ket{ \Psi^{+} }^{AB} = \frac{1}{\sqrt{2}} ( \ket{01}^{AB} + \ket{10}^{AB} )  \] 
and 
\[ \ket{ \Phi^{+} }^{AB} =  \frac{1}{\sqrt{2}} ( \ket{01}^{AB} + \ket{10}^{AB} )  \]
and the Pauli operation  $X= \begin{pmatrix} 0 & 1 \\ 1 & 0 \end{pmatrix} $.  

Let $\varphi^{AB}$ be the noumenal-phenomenal epimorphism on the global system $S = AB$ and let $\varphi^{A}$ be the noumenal-phenomenal epimorphism on system $A$.
We shall  prove that neither $\varphi^{A}$ nor $\varphi^{AB}$ are injective.  From the fact that $\varphi^{A}$ is not injective, it follows that for an arbitrary non-trivial system $A$, its associated noumenal-phenomenal epimorphism is not injective.  From the fact that $\varphi^{AB}$ is not injective, it follows that the noumenal-phenomenal epimorphism on the global systen is not injective.  Combining these two facts, we obtain that on any system distinct from the empty system, the noumenal-phenomenal epimorphism associated with that system is not injective. This proves the desired result.

The idea behind the proof that neither $\varphi^{A}$ nor $\varphi^{AB}$ is injective goes  as follow: 

Given that $AB$ is in   Bell state $\ket{\Psi^{+} }^{AB}$, applying Pauli operation $X$ to $A$ leads to phenomenal changes to  $AB$ but no phenomenal change to  $A$. It follows that the operation caused an underlying noumenal change in $AB$.  By locality, if the action of applying  $X$ to$A$ caused an underlying in $AB$ noumenal change to system $AB$,  the noumenal state of system $A$ has changed.  Thus applying the Pauli operation $X$ to system $A$ changes the  noumenal state of $A$ without changing the phenomenal state of  $A$. Hence $\varphi^{A}$ is not injective.

Next, applying operation $X$ to both system $A$ and system $B$ to Bell state $\ket{\Psi^{+}}$, leads to no phenomenal change on system $AB$.  But the effect of this action on system $A$ has been the application of $X$, which changed its  noumenal state, as previously argued.  Since we changed the noumenal state of system $A$, we changed the noumenal state of system $AB$.  Therefore, by applying Pauli operation $X$ on both system $A$ and $B$ to Bell state $\ket{\Psi^{+}} $, we changed the noumenal state of system $AB$ without changing its corresponding phenomenal state. This implies that $\varphi^{AB}$ is not injective.

Thus  in quantum theory, we may change the noumenal state of any (non-empty) system without changing its corresponding phenomenal state.

That is the idea behind the proof. Here  are the details.

By the surjectivity of the noumenal-phenomenal epimorphism $\varphi^{AB}$ on the global system $S = AB$ there exist a noumenal state $\langle \Psi^{+} \rangle^{AB} $ such that
\[ \varphi^{AB} ( \langle  \Psi^{+} \rangle^{AB}   ) = \ket{ \Psi^{+} } \!  \bra{ \Psi^{+} }^{AB} \, .  \]
Let us define 
\[ \langle \Phi^{+} \rangle^{AB} \isdef   ( X \otimes I ) \langle  \Psi^{+} \rangle^{AB}  \, . \]
From this, it follows that :
\[   \varphi^{AB} ( \langle  \Phi^{+} \rangle^{AB}   ) = \ket{ \Phi^{+} } \! \bra{ \Phi^{+} }^{AB}  \, ,  \]
since:
\begin{align*}
 \varphi^{AB} ( \langle \Phi^{+} \rangle^{AB}   ) =& \varphi^{AB} ( ( X \otimes I ) \langle  \Psi^{+} \rangle^{AB}  ) \\
                                                                               =&  ( X \otimes I ) \cdot  \varphi^{AB} ( \langle \Psi^{+} \rangle^{AB} )  \\
                                                                               =& (X \otimes I ) \cdot   \ket{ \Psi^{+} } \!  \bra{ \Psi^{+} }^{AB} \\
                                                                                = &(  X \otimes I ) \ket{ \Psi^{+} } \!  \bra{ \Psi^{+} }^{AB} ( X \otimes I) \\
                                                                               =&  \ket{ \Phi^{+} } \! \bra{ \Phi^{+} }^{AB}  \, . 
\end{align*}

The noumenal state
 \[    \langle  \Psi^{+} \rangle^{AB}  = \langle \Psi^{+} \rangle^{A} \odot \langle \Psi^{+} \rangle^{B}  \]
can be written thus as a product state where $\langle \Psi^{+} \rangle^{A} = \mathrm{tr}_{B} \left(  \langle \Psi^{+} \rangle^{AB} \right) $ and $ \langle \Psi^{+} \rangle^{B} = \mathrm{tr}_{A} \left(  \langle \Psi^{+} \rangle^{AB} \right) $. Here,   $\mathrm{tr}_{A}$ and $\mathrm{tr}_{B}$ are noumenal traces.

This implies that
\[ \langle \Phi^{+} \rangle^{AB}  = X \langle \Psi^{+} \rangle^{A}  \odot \langle \Psi^{+} \rangle^{B} \, ,  \]
since
\begin{align*}
\langle \Phi^{+} \rangle^{AB}    =&  ( X \otimes I ) \langle  \Psi^{+} \rangle^{AB}   \\
                                                                               =& (  X \otimes I ) ( \langle \Psi^{+} \rangle^{A} \odot \langle \Psi^{+} \rangle^{B} ) \\
                                =& X \langle \Psi^{+} \rangle^{A} \odot I \langle \Psi^{+} \rangle^{B} \\ 
                                                          = & X \langle \Psi^{+} \rangle^{A} \odot \langle \Psi^{+} \rangle^{B}  \, . 
\end{align*}

To prove that noumenal-phenomenal epimorphism $\varphi^{A}$ is not injective,  we show that $\varphi^{A} ( \langle \Psi^{+} \rangle^{A} ) $ and $ \varphi^{A} ( X \langle \Psi^{+} \rangle^{A}  ) $  are equal but that $ \langle \Psi^{+} \rangle^{A} $ and $X \langle \Psi^{+} \rangle^{A}$ are not equal.

First, we see that
\[\varphi^{A} (   \langle \Psi^{+} \rangle^{A}  ) = \varphi^{A} (  X \langle \Psi^{+} \rangle^{A} ) \, , \]
 since
\begin{align*}
 \varphi^{A} (  \langle \Psi^{+} \rangle^{A} )  &=  \varphi^{A} ( \mathrm{tr}_{B} ( \langle \Psi^{+} \rangle^{AB} )) \\
                                                                             &=   \mathrm{tr}_{B} ( \varphi^{AB} ( \langle \Psi^{+} \rangle^{AB} )) \\
                                                                            &= \mathrm{tr}_{B} ( \ket{ \Psi^{+} } \! \bra{ \Psi^{+} }^{AB} ) \\
                                                                            &= \frac{1}{2} \begin{pmatrix} 1 & 0 \\ 0 & 1 \end{pmatrix}    \\
                                                                           &= \frac{1}{2} I^{A}
\end{align*}
and
\begin{align*}
 \varphi^{A} ( X \langle \Psi^{+} \rangle^{A} )  &=  X \cdot  \varphi^{A} ( \langle \Psi^{+} \rangle^{A} ) \\
                                                                              &=  X \cdot  ( \frac{1}{2} I^{A} )   \\
                                                                                & = X  ( \frac{1}{2}  I^{A} ) X   \\             
& = \frac{1}{2} X^{2} \\
                                         &=  \frac{1}{2} I^{A}  \, . 
\end{align*}

Now we show that $ \langle \Psi^{+} \rangle^{A}$ and $ X \langle \Psi^{+} \rangle^{A} $ are not equal.
Recall that  $ \langle \Psi^{+} \rangle^{AB} $ and $ \langle \Phi^{+} \rangle^{AB} $ are distinct noumenal states, since they give rise to distinct corresponding phenomenal states.
From the fact that these two noumenal states are distinct and that $ \langle \Psi^{+} \rangle^{AB} = \langle \Psi^{+} \rangle^{A} \odot \langle \Psi^{+} \rangle^{B} $ and $\langle \Phi^{+} \rangle^{AB} = X \langle \Psi^{+} \rangle^{A} \odot \langle \Psi^{+} \rangle^{B}$, it follows that $ \langle \Psi^{+} \rangle^{A} $ and  $X \langle \Psi^{+} \rangle^{A}$ are distinct.  

We have now proved that the noumenal-epimorphism  $\varphi^{A} $ on system $A$ is not injective. It remains to prove that the noumenal-phenomenal epimorphism $\varphi^{AB}$ on system $AB$ is not injective.

To do so, we show that $ \langle \Psi^{+} \rangle^{AB} $ and $(  X \otimes X ) \langle \Psi^{+} \rangle^{AB} $ are distinct noumenal states giving rise to the same corresponding phenomenal state. From this it follows immediately that the noumenal-phenomenal epimorphism $\varphi^{AB}$ on the global system is not injective.

Let us  verify that they give rise to the same phenomenal  state:
\begin{align*}
 \varphi^{AB} ( ( X \otimes X ) \langle \Psi^{+} \rangle^{AB} ) &= ( X \otimes X ) \cdot \varphi^{AB} ( \langle \Psi^{+} \rangle^{AB} ) \\
                                                                                                   &= ( X \otimes X ) \cdot  | \Psi^{+} \rangle \! \langle \Psi^{+} |^{AB} \\
                                                                                                    &=  ( X \otimes X )| \Psi^{+} \rangle \! \langle \Psi^{+} |^{AB} ( X \otimes X ) \\
                                                                                                     & =| \Psi^{+} \rangle \! \langle \Psi^{+} |^{AB} \\
                                                                                                  & = \varphi^{AB} ( \langle \Psi^{+} \rangle^{AB}  )  \, .
\end{align*}

Now, we show that $ \langle \Psi^{+} \rangle^{AB}$ and $(  X \otimes X ) \langle \Psi^{+} \rangle^{AB} $ are distinct.  
Recall that 
\[ \langle \Psi^{+} \rangle^{AB} = \langle \Psi^{+} \rangle^{A} \odot \langle \Psi^{+} \rangle^{B} \, , \] 
and
 \[  ( X \otimes X ) \langle \Psi^{+} \rangle^{AB} = X \langle \Psi^{+} \rangle^{A} \odot X \langle \Psi^{+} \rangle^{B} \, . \]

From theorem  \ref{thm:uniquedec} and the previously proven fact that 
\[ \langle  \Psi^{+} \rangle^{A} \neq X \langle \Psi^{+} \rangle^{A} \, ,  \]
it follows that
\[  \langle \Psi^{+} \rangle^{A} \odot \langle \Psi^{+} \rangle^{B}  \neq  X \langle \Psi^{+} \rangle^{A} \odot X \langle \Psi^{+} \rangle^{B} \, . \]
And thus
\[  \langle \Psi^{+} \rangle^{AB} \neq (  X \otimes X ) \langle \Psi^{+} \rangle^{AB} \, .  \]

It follows that the noumenal-phenomenal epimorphism on the global system is not injective. This completes the proof.

\section{Change of basis}\label{basischange}

We state a few results related to evolution spaces over the same systems but in different bases.  

\begin{theorem}[Relation between bases]
Let  $A$ be a system, let $W$ be an operation on the global system $S$, and let $ \left\{ | i \rangle \right\}_{i \in I} $, $ \left\{ \ket{ k' } \right\}_{k' \in K}  $  be two orthonormal bases of $\mathcal{H}^{A}$ the Hilbert space associated to system $A$,
 \[  \left[ W \right]^{A}_{k'  \ell'  } = \sum\limits_{i \in I, j \in I} \langle k'  | i \rangle  \left[ W \right]_{ i j }^{A} \langle j |  \ell '  \rangle \, .    \]
\end{theorem}

\begin{proof}
\begin{align*}
 & \lbrack W \rbrack^{A}_{ k'  \ell'  } \\
= &  W^{\dagger} \left(  | \ell '  \rangle \! \langle k'  | \otimes I^{\overline{A}} \right) W \\
= & \sum\limits_{i,j} W^{\dagger}  \left( \ket{j} \langle j | \ell '  \rangle \langle k'  | i \rangle \bra{i} \otimes I^{\overline{A}} \right) W \\
= &  \sum\limits_{i,j} \langle k'  | i \rangle W^{\dagger}  \left( \ket{j} \! \bra{i} \otimes I^{\overline{A}} \right) W   \langle j | \ell'  \rangle \\  
=& \sum\limits_{i,j} \langle k' | i \rangle  \left[ W \right]_{ i j }^{A} \langle j |  \ell'  \rangle \, .  \\
\end{align*}
\end{proof}

\paragraph{Change of basis.} Let  $A$ be a system,  let $\mathcal{B}_{1} =  \{ | i \rangle \}_{i \in I}$ and $\mathcal{B}_{2} = \{ | k' \rangle \}_{k' \in K} $ be orthonormal bases of $\mathcal{H}^{A}$ the Hilbert space associated to system $A$, and  let $M^{A}$ be an element of $ \textsf{Evolution-Space}^{A}_{\mathcal{B}_{1} }$.
We define $\mathcal{B}_{2} \leftarrow \mathcal{B}_{1} ( M^{A}  )$ as the operator matrix in basis $\mathcal{B}_{2} $ whose $ k' \ell' $ entry is
\[ ( \mathcal{B}_{2} \leftarrow \mathcal{B}_{1} ( M^{A} ) )_{ k' l'  }  \stackrel{\text{def}}{=}  \sum\limits_{i \in I, j \in I} \langle k' | i \rangle  M^{A}_{ij} \langle j |  \ell' \rangle   \, .  \]

\begin{theorem} Let $A$ be a system, let $\mathcal{B}_{1}$ and $\mathcal{B}_{2}$ be orthonormal bases of $\mathcal{H}^{A}$, for any unitary operation $W$ on the global system,
\[ \mathcal{B}_{2} \leftarrow \mathcal{B}_{1} ( \lbrack W \rbrack^{A} _{\mathcal{B}_{1} } ) = \lbrack W \rbrack^{A}_{\mathcal{B}_{2}} \, . \]
\end{theorem}
\begin{proof}
\[  (  \mathcal{B}_{2} \leftarrow \mathcal{B}_{1} ( \lbrack W \rbrack^{A} _{\mathcal{B}_{1} } ) )_{k' l' } 
=  \sum\limits_{i, j} \langle k' | i \rangle  \lbrack W \rbrack^{A}_{ij} \langle j |  \ell' \rangle     = \lbrack W \rbrack_{k'  \ell' }^{A}  \, . \]
\end{proof}

By the previous theorem it follows that the image of an element of $ \textsf{Evolution-Space}^{A}_{\mathcal{B}_{1} }$   under the  function $ \mathcal{B}_{2} \leftarrow \mathcal{B}_{1} $ is an element of $ \textsf{Evolution-Space}^{A}_{\mathcal{B}_{2} }$.  Formally,
\[ \mathcal{B}_{2} \leftarrow \mathcal{B}_{1} \colon  \textsf{Evolution-Space}^{A}_{\mathcal{B}_{1} } \to  \textsf{Evolution-Space}^{A}_{\mathcal{B}_{2}} \, .  \]

For the following three corollaries, let $A$ be a system and let $\mathcal{B}_{1}$, $\mathcal{B}_{2}$ and $\mathcal{B}_{3}$ be orthonormal bases associated with it.

\begin{corollary} Function
$ \mathcal{B}_{1} \leftarrow \mathcal{B}_{1} $ is the identity  on  $\textsf{Noumenal-Space}^{A}_{\mathcal{B}_{1} } $.
\end{corollary}
\begin{proof} The proof is left to the reader.
\end{proof}

\begin{corollary}
$  ( \mathcal{B}_{3} \leftarrow \mathcal{B}_{2} ) \circ ( \mathcal{B}_{2} \leftarrow \mathcal{B}_{1} )  = \mathcal{B}_{3} \leftarrow \mathcal{B}_{1} $.
\end{corollary}
\begin{proof} The proof is left to the reader.
\end{proof}

\begin{corollary} $ \mathcal{B}_{2} \leftarrow \mathcal{B}_{1}  \colon \textsf{Noumenal-Space}^{A}_{\mathcal{B}_{1} } \to \textsf{Noumenal-Space}^{A}_{\mathcal{B}_{2} } $ is a bijection.
\end{corollary}
\begin{proof}
The proof is left to the reader.
\end{proof}

\end{document}